\documentclass[a4paper,11pt]{article}
\usepackage[english]{babel}
\usepackage[ansinew]{inputenc}
\usepackage{fontenc}
\usepackage{amsfonts}
\usepackage{amsmath}
\usepackage{amsthm}
\usepackage{enumerate}
\usepackage{textcomp}
\usepackage{geometry}
\usepackage{mathrsfs}
\usepackage{natbib}
\usepackage{dsfont} 
\usepackage{graphicx}
\usepackage{multirow}
\usepackage[ruled]{algorithm2e}

\usepackage[colorlinks=true,linkcolor=red,citecolor=blue]{hyperref}

\numberwithin{equation}{section}

\newcommand{\ind}[1]{\ensuremath{\mathds{1}_{#1}}}

\newtheorem{definition}{Definition}[section]

\newtheorem{lemma}{Lemma}[section]
\newtheorem{proposition}{Proposition}[section]

\theoremstyle{definition}
\newtheorem{remark}{Remark}[section]
\newtheorem{example}{Example}[section]

\begin{document}

\author{Diaa Al Mohamad\footnote{Corresponding author. E-mail: diaa.almohamad@gmail.com} \;\; \; \;\; Erik W. van Zwet \;\;\;\;\; Jelle J. Goeman \;\; \; \;\; Aldo Solari
 \\ \normalsize{Leiden University Medical Center, The Netherlands}\\
  \normalsize{University of Milano-bicocca, Italy} 
\\
Last update 
}

\title{Simultaneous confidence sets for ranks using the partitioning principle - Technical report}

\maketitle

\begin{abstract}
Ranking institutions such as medical centers or universities is based on an indicator accompanied with an uncertainty measure such as a standard deviation, and confidence intervals should be calculated to assess the quality of these ranks. We consider the problem of constructing simultaneous confidence intervals for the ranks of centers based on an observed sample. We present in this paper a novel method based on multiple testing which uses the partitioning principle and employs the likelihood ratio (LR) test on the partitions. The complexity of the algorithm is super exponential. We present several ways and shortcuts to reduce this complexity. We provide also a polynomial algorithm which produces a very good bracketing for the multiple testing by linearizing the critical value of the LR test. We show that Tukey's Honest Significant Difference (HSD) test can be written as a partitioning procedure. The new methodology has promising properties in the sens that it opens the door in a simple and easy way to construct new methods which may trade the exponential complexity with power of the test or vice versa. In comparison to Tukey's HSD test, the LR test seems to give better results when the centers are close to each others or the uncertainty in the data is high which is confirmed during a simulation study. 
\end{abstract}
\section{Introduction}
Performance indicators are increasingly used by the public sector to conclude some information about the performance of institutions such as medical centers, universities or social services. The interest in these indicators has increased during the past three decades. Intuitively, whenever one wants to book at a hotel or go to a hospital, he would aim for the best in some sense or at least make sure that the closest for example hospital is a good option and not ranked the worst among the available list of options. Besides, it is important for example for the government to know and decide towards which institutions (medical or educational) to direct its funds, and detect and be alarmed about those which have poor performance. As is generally the case in statistics, providing a numerical indicator about the performance or the degree of goodness of some institution is not sufficient because of the different variations and factors that influence on this indicator. Thus, a standard deviation is generally provided with these indicators. For a general reader, reading tables of performances with standard deviations is not a simple task or easily interpretable. Providing the ranks of the institution has a more intuitive and an easier aspect for the general reader and non specialists. It is however not straightforward to change performance indicators provided with their standard deviations into ranks especially if we want to produce a result with a simultaneous confidence for all the centers at a time. Goldstein and Speigelhalter \cite{GoldsteinSpiegel} discussed the importance of introducing performance indicators to compare institutions and addressed the statistical issues related to this problem. They point out that performance indicators lead in a natural way to ranking. They also emphasize on the need for "interval estimation" in order to detect the uncertainty accompanying these ranks. Nowadays, rankings of universities and medical institutions are being published regularly, but associated confidence intervals for their ranks are wide enough to make it difficult to draw a conclusion over the ranks of these institution, see \cite{Spiegelhalter}. Less effort is devoted in the literature towards improving these rankings by providing shorter confidence intervals for the same confidence level by either improving on the indicator used or on the methodology which produces the confidence intervals. In this paper, we are only interested in the latter.\\
More formally, Let $\mu_1,\cdots,\mu_n$ be $n$ real valued centers. Denote $r_1,\cdots,r_n$ the ranks of these centers respectively. When the centers are all different, the ranks are calculated by counting down how many centers are below the current center. When there are ties between the centers, we suppose that each of the tied centers possesses a set of ranks. In other words, assume that we have only 3 centers $\mu_1,\mu_2$ and $\mu_3$ such that $\mu_1=\mu_2<\mu_3$. Then, the rank of $\mu_1$ is the set $\{1,2\}$ and the rank of $\mu_2$ is also the set $\{1,2\}$, whereas the rank of $\mu_3$ is 3. A more formal definition is introduced in paragraph \ref{sec:HowToRank} hereafter.\\ 
Let $\mathcal{Y}=\{y_1,\cdots,y_n\}$ be a sample drawn from the Gaussian distributions
\begin{equation}
y_i \sim \mathcal{N}(\mu_i,\sigma_i^2), \quad \text{for } i\in\{1,\cdots,n\},
\label{eqn:TheGaussModel}
\end{equation}
where the standard deviations $\sigma_1,\cdots,\sigma_n$ are known whereas the centers $\mu_1,\cdots,\mu_n$ are unknown. We call the ranks induced from this observed sample the \emph{empirical ranks}\footnote{This is different from the empirical ranks in Bayesian approaches which estimate the ranks and call these estimates as empirical ranks.}. These ranks might be different from the true ranks of the centers. We aim on the basis of this sample to construct simultaneous confidence intervals for the ranks of the centers. In other words, for each $r_i$ we search for a confidence interval $[r_{i,L}(\mathcal{Y}),r_{i,U}(\mathcal{Y})]$ such that:
\[\mathbb{P}\left(r_i\in [r_{i,L}(\mathcal{Y}),r_{i,U}(\mathcal{Y})],\forall i\in\{1,\cdots,n\}\right)\geq 1-\alpha\]
for a pre-specified confidence level $1-\alpha$. Of course, when the rank $r_i$ is a set of ranks, then the signe $\in$ is replaced by an inclusion $\subset$.\\
In the literature, the ranking problem was considered in several papers from statistics and other fields especially the medical one. Several methods were introduced in order to create confidence intervals for the ranks, but generally no global comparison between these methods is available because of the vast statistical approaches and options. We mention four major approaches from the literature. The first one is Bootstrap-based (or frequentist) approaches, see \cite{GoldsteinSpiegel}, \cite{Spiegelhalter}, \cite{Zhang}, \cite{HallMillerInconsistRank}, \cite{XieMiddleRank}, \cite{Gerzoff}, \cite{Feudtner} among others. The second one are Bayesian approaches \cite{LairdThomasBayes}, \cite{LinThomasBayes}, \cite{LinThomasBayesBis}, \cite{NomaBayes}, \cite{LingsmaER}, \cite{Hans} among others. The third one is based on testing pairwise differences between the centers with (partial) or without correction for multiple testing, see \cite{LemmersZscore}, \cite{LemmersZscoreAgain}, \cite{HolmUnpublished}, \cite{TingBieMasterThesis} and \cite{OurTukeyPaper}. We mention the interesting paper \cite{Rafter} which summarizes well-known tests concerning multiple comparisons and detail their properties although it is not oriented towards producing ranks, see also the briefer paper \cite{KimMCMsReview}. The fourth approach is funnel plots which are used to detect divergence from average behavior of the compared centers using a $2\sigma$ ($95\%$ confidence interval) and a $3\sigma$ ($99.8\%$ confidence interval) rules, see \cite{TekkisFunnelPlot}, \cite{SpiegelhalterFunnelPlots} among others. We may also mention other works which use the parametric model (\ref{eqn:TheGaussModel}) to produce confidence intervals of some comparative score without truly producing ranks, see \cite{ParryWscore}. The state of the art about the topic is very vast, and this is only a subset of the published papers, and we hope this list provides an idea about the different approaches considered in the literature about the topic.\\

Bootstrap-based methods although being the most used in practice to build confidence intervals because of their intuitive and easy idea, they were criticized for their inability to handle (near) ties between the centers, see for example \cite{XieMiddleRank} and \cite{HallMillerInconsistRank}. Bayesian approaches focus generally on producing ranks estimates and have very little interest in producing (simultaneous) credible sets for them. Moreover, generally in Bayesian methods a continuous (Gaussian for example) prior on the centers is used. This \emph{implicitly} assumes that the probability of ties between the centers is zero. Funnel plots do not produce confidence intervals for each center, but a general idea about the "normal" behavior of a center and produce somewhat an "alert" and "alarm" zones for the behavior of the institutions. Moreover, approaches based on multiple comparisons, which are the least used in the literature to produce ranks, are either partially corrected for multiple testing or not at all. Thus, they do not really produce simultaneous confidence intervals. Note that in these approaches, and to the best of our knowledge, \emph{simultaneous} confidence intervals are never discussed except for the paper of \cite{Zhang} which is based on Bootstrapping and the paper of \cite{OurTukeyPaper} which is based on multiple comparison testing. \\
 
In this paper, we will see how we may use pairwise comparisons to produce simultaneous confidence intervals for the ranks by employing Tukey's HSD test. We will also present a new method which employs multiple testing and is based on the partitioning principle. The new procedure is based on the likelihood ratio (LR) test which was \emph{never} used or considered in the literature before for this particular purpose of producing confidence intervals for ranks. The LR test was only used in similar contexts to test hypotheses related to some ordering between the centers against all the alternatives (\cite{Robertson78}, \cite{BarlowBook}) or even to test equality between all the centers against some ordering over the centers (\cite{BartholomewPAVA}, \cite{Robertson78}). The two procedures, Tukey's HSD and the LR, presented in this paper have different power properties. While Tukey's procedure aims at identifying significant differences between two centers, our LR-based procedure aims at identifying differences within sub-groups of centers by depending on the (possibly small) contributions produced by each of the centers inside these sub-groups. The simulations at the end of this paper show cases where each approach outperforms the other showing that there is no uniformly winner and giving the interested reader the choice according to his point of view about the application and the goal of his study.

The paper is organized as follows. In Section \ref{sec:RanksDef}, we give a formal definition of the ranking problem and set the objectives. In Section \ref{sec:PartitionPrincip}, we present the partitioning principle and explain how to use it in order to produce confidence intervals for the ranks. In Section \ref{sec:LRTest}, we present the likelihood ratio (LR) test for ordered hypotheses (against all alternatives). In Section \ref{sec:shortcuts}, we give an idea about the complexity of the partitioning procedure and introduce some shortcuts which greatly simplify the complexity. Section \ref{sec:Algos} is devoted for practical issues where two algorithms are introduced to perform the partitioning procedure. In Section \ref{sec:ApproxPartition}, we introduce a way to approximate the partitioning result with a polynomial algorithm. In Section \ref{sec:TukeyHSD}, we present briefly other methods from the literature and focus on Tukey's honest significant test. Finally, in Section \ref{sec:simulations}, we study the results of the LR test and Tukey's HSD on simulated samples. Software to perform the methods presented in this paper are available in the \texttt{ICRanks} package downloadable from CRAN.
\section{Simultaneous confidence intervals for ranks}\label{sec:RanksDef}
Consider the sample $\mathcal{Y} = \{y_1,\cdots,y_n\}$ distributed independently from the Gaussian distributions $\mathcal{N}(\mu_i,\sigma_i^2)$ where $(\mu_1,\cdots,\mu_n)\in\mathbb{R}^n$ is the vector of unknown true centers. The standard deviations $\sigma_1,\cdots,\sigma_n$ are considered known.
\begin{definition}[rank of a center in the presence of ties] \label{def:Ranks}
Let $\mathcal{S}_n$ be the set of permutations from $\{1,\cdots,n\}$ to itself. We define the ranking function
\begin{eqnarray}
R&:& \mathbb{R}^n \rightarrow \mathcal{P}\left(\mathcal{S}_n\right) \nonumber \\
R(\mu_1,\cdots,\mu_n) &=& \{p\in\mathcal{S}_n: \; \mu_{p(1)}\leq \cdots \leq \mu_{p(n)}\}
\label{eqn:RankDef}
\end{eqnarray} 
The rank of center $\mu_i$ is then given by the set $\{p(i),\; p\in R(\mu_1,\cdots,\mu_n)\}$.
\end{definition}
This definition reads as follows. The ranks of centers $\mu_1,\cdots,\mu_n$ are the set of \emph{all} permutations of $\mathcal{S}_n$ such that the centers permute with each others through them. This is to emphasize on the fact that not any combination of permutations produce a valid set of ranks for the centers. A \emph{valid set of ranks} must verify the following condition. Let $R_i(\mu_1,\cdots,\mu_n) = \{a_{1,i},\cdots,a_{k,i}\}$ be the sorted set of ranks obtained from the permutations corresponding to center $\mu_i$. In order for the set $R_i(\mu_1,\cdots,\mu_n)$ to be a valid set of ranks, all natural numbers between $a_{1,i}$ and $a_{k,i}$ must be included in the set $R_i(\mu_1,\cdots,\mu_n)$. In other words, it must not skip any rank between the lowest and the highest admissible rank for center $\mu_i$. This is the reason why in formula (\ref{eqn:RankDef}), \emph{all} permutations verify $\mu_{p(1)}\leq \cdots \leq \mu_{p(n)}$ must be included so that they constitute a valid set of ranks.\\
Since the observed values $y_1,\cdots,y_n$ are modeled using a Gaussian distribution, the empirical ranking $R(y_1,\cdots,y_n)$ contains only one permutation. Using Definition \ref{def:Ranks}, the objective of this paper becomes the search for a subset $S(\mathcal{Y})$ of $\mathcal{P}(\mathcal{S}_n)$ implying valid ranks such that
\[\mathbb{P}\left(R(\mu_1,\cdots,\mu_n)\subset S(\mathcal{Y})\right) \geq 1-\alpha.\]
Thus, we are looking for a confidence set which contains the set-ranks with probability at least $1-\alpha$. In the literature, we talk about confidence intervals for the ranks instead of confidence sets, so that we continue to adopt this notation instead of talking about a confidence sets. Besides, as we talk about a confidence interval for the ranks, we avoid the subtlety whether a confidence set of ranks is a valid set of ranks or not. An interval (in the discrete space of ranks) by definition contains all ranks between the lowest and highest point in it.\\
We propose to use multiple testing techniques in order to construct the set of simultaneous confidence intervals. Therefore, it is important to start by precising the set of elementary hypotheses we need to test. Since we are looking to infer about the ranks, let $H_I:R(\mu_1,\cdots,\mu_n) = \{p_i,i\in I\}$ for $p_i\in \mathcal{S}_n$ and some subset $I$ from $\mathcal{P}(\{1,\cdots,n\})$ such that $\{p_i,i\in I\}$ is a valid set of ranks\footnote{If it is not a valid set of ranks, there is no meaning in testing it because the equality $R(\mu_1,\cdots,\mu_n) = \{p_i,i\in I\}$ will never hold.}. The set of elementary hypotheses contains then all possibilities for the set $I$ in  $\mathcal{P}(\{1,\cdots,n\})\}$ which result in valid sets of ranks. \\
In the literature on multiple testing, it is well-known that when we test jointly several hypotheses, we need to count for the inflammation of type I error and correct the critical level, see \cite{GoemanSolariMultHypoGenom}. Several methods exist in the literature to correct for multiple testing. The partitioning principle (\cite{Finner}, \cite{GoemanSeqRej}) is considered as a powerful method to correct for type I error and is considered superior to most other existing methods which is the motivation why we propose to use it here.
\begin{remark}
In this paper, we adapt the concept of mapping tied centers into intervals of ranks instead of a single rank. It is also possible to map equal ranks to a single value; the middle rank, see \cite{LairdThomasBayes} and \cite{XieMiddleRank}. If the 1st and the 2nd centers are equal, then each one of them get a middle rank equal to $\frac{1}{2}$. The ranks are then defined by (\cite{XieMiddleRank})
\[r_i = 1+\sum_{j\neq i}{\ind{\mu_i = \mu_j}} + \frac{1}{2}\sum_{i\neq j}{\ind{\mu_i = \mu_j}}.\]
This coincides with equation (\ref{eqn:RankDef}) when there are no ties. Xie et al. consider continuous estimators for these ranks which in a special case appear as if we are summing up p-values of all pairwise comparisons of the center $\mu_i$ w.r.t. other centers, see also \cite{LingsmaER}. The interpretation of middle ranks and continuous estimators is not easy especially for non-specialists. Moreover, although it is possible to produce integer ranks based on estimated (continuous) ranks by ranking them again, related confidence intervals are significantly wider than those produced for the estimated (continuous) ranks, see \cite{XieMiddleRank}. Therefore, this approach is not considered or pursued in this paper.
\end{remark}

\section{The partitioning principle}\label{sec:PartitionPrincip}
The main idea in the partitioning principle (\cite{Finner}) is to partition the union of the hypotheses of interest into disjoint sub-hypotheses or partitions such that each hypothesis can be represented as the disjoint union of some of these partitions. We then reject an elementary hypothesis if \emph{all} sub-hypotheses included in it are rejected, otherwise it will not be rejected.\\
In other words, let $\mathcal{H}$ be the set of hypotheses that we need to reject. Define the set of all sub-hypotheses, say $\mathcal{P}$. For any hypothesis $H\in\mathcal{H}$, we must be able to find a set of sub-hypotheses $\{H_i, i\in I\}$ such that $H\subset \cup_{i\in I} H_i$. The inclusion here means that for all $\phi\in H$, there exists $i\in I$ such that $\phi\in H_i$. Moreover, any two sub-hypotheses must be disjoint, that is $\forall H, K\in\mathcal{P}, H\cap K=\emptyset$.\\
How does it work? In order to reject a hypothesis $H$, we look at all sub-hypotheses from $\mathcal{P}$ which have an intersection with $H$. If they are all rejected at level $\alpha$, then $H$ is rejected at level $\alpha$ too. Otherwise, $H$ is not rejected. \\
The partitioning principle ensures familywise error (FWER) control at level $\alpha$, see \cite{GoemanSeqRej}. The key idea is that since the tested partitions are disjoint, at most one of them is true. Therefore, even though no multiplicity adjustment to the levels of the tests is made, the probability of rejecting at least one true null hypothesis is at most $\alpha$.\\
There are several ways to construct the partitioning of the parameter space, but we need to choose the one which suits us the best in terms of power and execution time. The latter may be a very important factor. Indeed, producing the set of partitions can easily make the number of hypotheses to be tested grow exponentially with the dimension of the parameter space (the number of the centers). It is thus important to find a way to reduce this complexity by finding relations among tested partitions. In the literature, these relations are called shortcuts \cite{CalianShortcuts}. In our work, a very essential shortcut will be used in order to reduce the complexity of the algorithm from a super exponentially complex algorithm into an exponentially complex one, see Section \ref{sec:shortcuts}. Using also some simple tricks, we are able to test all the partitions in a reasonable time up to some size of the dataset (at least $n=50$), see Section \ref{sec:Algos}.\\
The Partitioning principle was first introduced by \cite{PartitioningPrinciple}. Finner and Strassburger \cite{Finner} proved that the Partitioning principle is at least as powerful as closed testing (\cite{MarcusClosedTesting}) and can be more powerful in some situations. Note that other existing methods which correct for multiple testing are still in the competition because it is not always (easy or) possible to generate a suitable partitioning\footnote{The same holds for the closed testing procedure.}. Besides, the complexity of a partitioning procedure is generally extremely high and without finding shortcuts which reduce this complexity to a reasonable one, the partitioning principle might lose its appeal against other fast correction methods.


\section{Producing confidence intervals for ranks with the partitioning principle}\label{sec:HowToRank}
The way we formulated our elementary hypotheses in Section \ref{sec:RanksDef} makes the elementary hypotheses themselves a partitioning of the parameter space. Indeed, the parameter space is the union of all possible ranks in the sens of Definition \ref{def:Ranks}, that is a subset of $\mathcal{P}(\mathcal{S}_n)$ the set of all subsets of the set of permutations. Any two elementary hypotheses are disjoint. Indeed, two different elementary hypotheses $H$ and $K$ must be different by at least one set-rank of a center, say $\mu_1$. For example, if $H$ says that $\mu_1$ has the set-rank $\{1,2,3\}$ whereas $K$ says that $\mu_1$ has the singleton set-rank $\{1\}$, then a vector $(\mu_1,\cdots,\mu_n)$ verifying the constraint of $H$ cannot verify the constraint of $K$ because $\mu_1$ cannot be equal to $\mu_2$ and strictly inferior to it in the same time. On the other hand, since we are considering all possible ranks, the set of partitions is equal to the parameter space, hence it constitutes a partitioning of it. Finally, it is the coarsest partitioning. This is immediate because we are considering all possible ranks as the partitioning.\\
Let us consider a very simple example where we have a set of three centers, namely $A,B$ and $C$. The elementary hypotheses are:
\[H_0 = \{(1,2,3)\} , H_1 = \{(1,3,2)\}, H_2 = \{(2,1,3)\}, H_3 = \{(3,2,1)\}, H_4 = \{(3,1,2)\}, H_5 = \{(2,3,1)\} \]
\[H_6 = \{(1,2,3),(1,3,2)\}, H_7 = \{(1,2,3),(2,1,3)\}, H_8 = \{(1,3,2),(3,1,2)\}, H_9 = \{(2,1,3),(2,3,1)\} \]
\[ H_{10} = \{(3,2,1),(3,1,2)\}, H_{11} = \{(3,2,1),(2,3,1)\}\]
\[H_{12} = \{(1,2,3),(1,3,2),(2,1,3),(3,2,1),(3,1,2),(2,3,1)\}.\]
We did not include all subsets of $\mathcal{S}_3$ simply because they do not form correct ranks. For example the set $\{(1,2,3),(3,2,1)\}$ gives center A the set of ranks $\{1,3\}$. This is not a valid set of ranks because it skips the middle rank, that is rank 2.\\
Each of the previous hypotheses can be represented in a more elegant and easy-to-read way using equality and inequality constraints. In figure (\ref{fig:Partitions3by3}), we illustrate all partitions for the three centers $A,B$ and $C$ using equality and inequality relations. This actually an equivalent way to represent the partitions without writing explicitly the set of corresponding permutations which will become unreadable as soon as the number of centers slightly increases. It is also easier to build a test for the partitions based on equality and inequality constraints than permutation constraints. The partitions as illustrated in figure (\ref{fig:Partitions3by3}) are as follows.
\begin{enumerate}
\item We have one partition defined with two equalities between the centers, namely $A=B=C$. This corresponds to hypothesis $H_{12}$;
\item We have 6 partitions defined with one inequality and one equality, namely:
\begin{eqnarray*}
A<B=C,\; A=B<C,\; B<A=C,\\
 B>A=C,\; C<A=B,\; A>B=C,
\end{eqnarray*}
which correspond to hypotheses $H_6,H_7,H_8,H_9,H_{10},H_{11}$.
\item We have 6 partitions defined with two inequalities, namely:
\begin{eqnarray*}
A<B<C,\; A<C<B,\; C<A<B, \\
B<A<C,\; B<C<A, \; C<B<A,
\end{eqnarray*}
which correspond to hypotheses $H_0,H_1,H_2,H_3,H_4,H_5$.
\end{enumerate}
\begin{figure}[ht]
\centering
\includegraphics[scale=0.55]{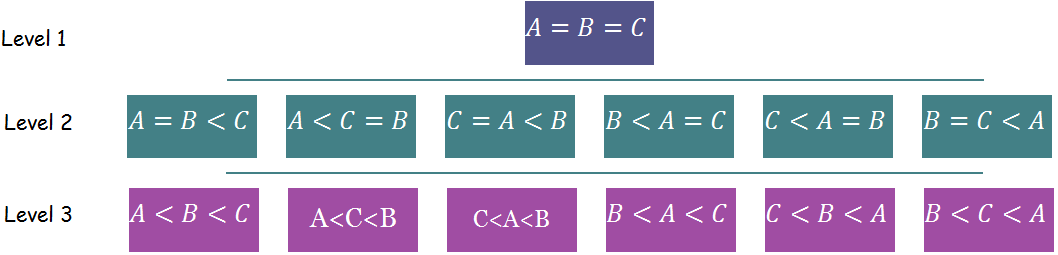}
\caption{Partitions with three individuals.}
\label{fig:Partitions3by3}
\end{figure}

In this paper we say that several hypotheses belong to the same level if they have the same number of equalities in their definition. In the 3-center example shown here above, the number of levels is 3. The first level is the top level and has only one hypothesis (partition) with equality between all the centers. The second level contains all partitions with one inequality (and one equality). Finally, the third level contains all partitions with two inequalities (no equalities).
\begin{definition}[Inclusion between hypotheses]
\label{def:InclusionConfigs}
We say that a hypothesis $H_1$ includes another hypothesis $H_2$, denoted using the usual notation $H_2\subset H_1$, if $H_2$ does not introduce any new equalities between the centers in comparison to $H_1$. In other words, $H_2\subset H_1$ if $R(H_2)\subset R(H_1)$ where $R(H)$ is the set of ranks of the centers when they follow the constraints defining the hypothesis $H$.
\end{definition}
The following result provides a way to deduce simultaneous confidence intervals at level for the ranks $1-\alpha$ based on testing the partitions at level $\alpha$. For the time being, we do not impose any condition on the statistical test apart from the fact that it must has a significance level equal to $\alpha$.
\begin{proposition}
The union of unrejected partitions at level $\alpha$ constitutes simultaneous confidence intervals for the ranks of the centers at level $1-\alpha$.
\end{proposition}
\begin{proof}
Since the partitioning principle ensures that the FWER is below $\alpha$, we may write
\[\mathbb{P}\left(\text{Number of type I errors} \geq 1\right) \leq \alpha\]
which is equivalent to 
\[1-\mathbb{P}\left(\text{Number of type I errors} = 0\right) \leq \alpha.\]
Thus the probability that no false rejection is made exceeds $1-\alpha$ and all rejections made by testing them at level $\alpha$ are true rejections with joint probability at least $1-\alpha$. Denote $\cup_{i\in I}P_i$ the set of rejected partitions at level $\alpha$ and $\mu_T$ the true vector of centers. We can write
\[\mathbb{P}\left(\mu_T\notin \cup_{i\in I}P_i\right) \geq 1-\alpha.\]
Since the set of partitions $P_i$ covers the parameter space and they are disjoint, then if the vector of centers is not in any of the rejected partitions, it is a fortiori in the unrejected set of partitions. This entails that
\[\mathbb{P}\left(\mu_T\in \mathbb{R}^n\setminus \cup_{i\in I}P_i\right) = 1-\alpha.\]
Finally, because this union of unrejected partitions implies a set of simultaneous confidence intervals for the ranks of the centers, this set has a confidence level of at least $1-\alpha$.
\end{proof}

\section{The likelihood ratio test}\label{sec:LRTest}
The partitioning principle is a general method to correct for multiple testing and needs a local test to test its partitions. Choosing the test is independent from the partitioning principle. This means that a bad choice of the statistical test will lead to wide confidence intervals and this does not contradict with the optimality and superiority properties of the partitioning principle announced in Section \ref{sec:PartitionPrincip}. The partitions are (null) hypotheses of the form $\mu_i<\mu_j=\mu_l< ... <\mu_m = \mu_k$ with equality or strict inequality relating the centers. In the literature on ordered hypotheses, there is not yet a general result about an optimal test in this context. However, as stated by Bartholomew \cite{BartholomewPAVA} about the use of the likelihood ratio test: "Although it is not clear, in this case, whether the test derived has optimum properties, the method has a strong intuitive appeal and leads to a meaningful test".\\
The log-likelihood is written for the Gaussian model as
\[L(\mu_1,\cdots,\mu_n) = -\frac{n}{2}\log(2\pi) - \sum_{i=1}^n{\log(\sigma_i)} - \frac{1}{2}\sum_{i=1}^n{\frac{(y_i-\mu_i)^2}{\sigma_i^2}},\]
and the log-likelihood ratio statistic is given by
\[LR = -2\left(\max_{\mu_1,\cdots,\mu_n\in H_0}L(\mu_1,\cdots,\mu_n) - \max_{\mu_1,\cdots,\mu_n\in \mathbb{R}^n} L(\mu_1,\cdots,\mu_n)\right).\]
The maximum of the log-likelihood function under no constraint on the centers (that is under the alternative) is attained for $\mu_i = y_i$. Hence,
\[LR = \min_{\mu_1,\cdots,\mu_n\in H_0}\sum_{i=1}^n{\frac{(y_i-\mu_i)^2}{\sigma_i^2}}.\]
Under the null hypothesis, the maximum likelihood can be either calculated analytically for some hypotheses or iteratively using the Pool Adjacent Violators Algorithm (PAVA). This depends on the ordering of the observed sample whether it follows the ordering imposed by the null hypothesis or not.\\
The most simple case to calculate is the top hypothesis with only equalities between the centers. The maximum of the log-likelihood function is attained on the average of the observations. Let $H_0:H_{1,s}<\cdots<H_{k,n}$ where $H_{ij}=\{\mu_i=\cdots=\mu_j\}$. In other words, we write the hypothesis $H_0$ as a union of groups of centers with equality among them. It can be shown that if $y_1<y_2<\cdots<y_n$, then the maximum likelihood is attained on $H_0$ and the LR is given by
\[LR = \min_{\mu_1,\cdots,\mu_n\in H_0}\sum_{i=1}^n{\frac{(y_i-\mu_i)^2}{\sigma_i^2}} = \sum_{i,j}\sum_{l: \mu_l\in H{i,j}}{\frac{(y_l-\hat{\mu}_{H_{i,j}})^2}{\sigma_l^2}},\]
where
\[\hat{\mu}_{H_{i,j}} = \frac{1}{\sum_{l: \mu_l\in H_{i,j}}{\frac{1}{\sigma_l^2}}} \sum_{s: \mu_s\in H_{i,j}}{\frac{y_s}{\sigma_s^2}}.\]
Notice that if $j=i$, that is if the set $H_{i,j}$ contains but one center, the argument of the maximum $\mu_i$ is equal to the single observation $y_i$ and so all subsets $H_{i,j}$ with a single center do not appear in the calculus of the maximum likelihood. This means that we only average inside the sets with equality between the centers.\\
The case when the observed values do not follow the ordering imposed by the null hypothesis is treated using the pool adjacent violators algorithm (PAVA). The PAVA is an iterative algorithm introduced independently by several authors such as Ayer et al. \cite{AyerPAVA}, van Eeden \cite{vanEeden} and Bartholomew \cite{BartholomewPAVA}. More details about the algorithm can be found in the review paper of Leeuw et al. \cite{LeeuwReview} and the book of Barlow et al. \cite{BarlowBook}. Function \texttt{isoreg} from \texttt{STATS} package in the statistical program \texttt{R} calculates the optimum of a function under full order restriction using the PAVA by taking as an input the maximum likelihood estimator inside each set of centers related by an equality (here it is the average).\\

\section{A note on the probability distribution of the likelihood ratio statistic under order constraints}\label{sec:CriticalLevel}
Hypotheses of the form $\mu_1\leq\cdots\leq\mu_n$ were considered in the literature as null or alternative hypotheses. The most easy-to-understand papers providing the probability distribution for the likelihood ratio statistic for tests related to such hypothesis are \cite{BartholomewPAVA} and \cite{Robertson78} among others. More general approaches can be found in \cite{Kudo} and \cite{Shapiro}. More details can also be found in the books of Barlow \cite{BarlowBook} and Silvapulle and Sen \cite{SilvapulleBook}. The case when some of the inequalities are replaced by equalities is discussed in \cite{Shapiro} in a very general context. We will review some of these results for the sake of completeness of the topic. More details about the calculus are given in appendix for the simple case of three centers, namely $A,B$ and $C$.\\
\paragraph{Under equality of all the centers.} This is the most simple situation. In this case, the LR statistic is given by
\[LR(\mu_1=\cdots=\mu_n) = \frac{1}{\sum_1^n{1/\sigma_i^2}}\sum_{j=1}^n{\frac{1}{\sigma_j^2}\left(y_j - \frac{1}{\sum_1^n{1/\sigma_i^2}}\sum_{k=1}^n{y_k/\sigma_k^2}\right)^2}.\]
The LR statistic is distributed as a $\chi^2(n-1)$.\\
\paragraph{Under strict inequalities between the centers.} Without any prior knowledge of the data and considering each hypothesis (configuration) separately, all configurations sharing the same number of equalities and inequalities are equivalent. We will pick the first one and give the probability distribution of the corresponding log-likelihood ratio (LR) statistics.  Corollary 2.6 in \cite{Robertson78} states the result in the general context of $n$ centers
\[\mathbb{P}_{\mu_1=\cdots=\mu_n}\left(LR(\mu_1<\cdots<\mu_n)>\gamma\right) = \sum_{l = 1}^{n-1}{P(l,n)\mathbb{P}\left(\chi^2(n-l)>\gamma\right)},\]
where the numbers $\{P(l,n),l=1,\cdots,n\}$ sum up to 1 and are functions of the standard deviations $\sigma_1,\cdots,\sigma_n$. The numbers $P(l,n)$ were first mentioned in the paper of \cite{BartholomewPAVA} where the author calculated these numbers up to $n=5$ with exact formulas and conjectured a recursive formula for any $n$. The formula was proved after that by \cite{Miles59} who also calculated these numbers in the equal-$\sigma$ case that is when all the centers have the same standard deviation. They are given by
\[P(l,n) = \frac{|S_n^l|}{n!}\]
and the numbers $|S_n^l|$ for $l=1,...,n$ are the unsigned Stirling numbers of the first kind (See \cite{Miles59}). \\
The calculus of the numbers $P(l,n)$ in the general case of different values for the $\sigma_i$'s becomes very complicated as $n$ increases. The R package \texttt{ic.infer} provides functionalities to calculate these numbers and to perform tests on hypotheses of the form $\mu_1<\cdots<\mu_n$ (even when some of the inequalities are replaced by equalities), see \cite{IcInferPackage}.
\paragraph{Under a mixture of equalities and inequalities.} Since all configurations are equivalent prior to the knowledge of observations, it suffices to calculate the probability distribution for the hypothesis $A_1 = A_2=\cdots=A_d<A_{d+1}<\cdots<A_{n}$. It suffices to notice that in comparison to the case of only inequalities, this is the same. It is as if we merged the first $d$ individuals in only one but with $d-1$ degrees of freedom. Now using the result from Robertson \cite{Robertson78} and by adjusting for the equalities in the degrees of freedom of the $\chi^2$ distribution, we may write:
\begin{multline*}
\mathbb{P}_{A_1=\cdots=A_n}\left(LR(A_1 = A_2=\cdots=A_d<A_{d+1}<\cdots<A_{n}) > \gamma\right) = \\ \sum_{l = 0}^{n-d}{P(l+1,n-d+1)\mathbb{P}(\chi^2(n-l-1)>\gamma)},
\end{multline*}
The case of $n=3$ and $d=2$ gives directly formula \ref{eqn:ProbaLawEqIneq3}. See also \cite{Shapiro} for a general result.
 
Notice that it is not possible to calculate the probability distribution of the LR except under the least favorable case, that is when all the means are equal, because otherwise we get non-central $\chi^2$ distributions with unknown parameters. As we mentioned in Section \ref{sec:HowToRank}, if we do not reject the top hypothesis $\mu_1=\cdots=\mu_n$, we get trivial confidence intervals for the ranks. Thus, going further in the partition scheme does not have any sens unless we have already rejected the least favorable case under which we are calculating the critical level for all other null hypotheses. It is of further interest then to find a better critical value which does not assume the least favorable hypothesis. Notice that the least favorable case is the top hypothesis in our partitioning scheme and exists in every hypothesis related to the ranks of the centers. Thus it must be rejected in the beginning because otherwise the confidence intervals for the ranks of the centers are simply $[1,n]$.

\subsection{Adjusting the critical value}\label{subsec:AldoAdjust}
It is interesting to try and find a way to deal with the partitions without the need to suppose being under the worst case scenario, that is when all the centers are equal.
\begin{proposition}
In order to test a partition defined through $k$ equalities and $n-k$ inequalities constraints using the LR test, it suffices to compare the LR statistic with the $\chi^2$ quantile of order $1-\alpha$ and with $k$ degrees of freedom.
\end{proposition}
\begin{proof}
Suppose we want to deal with the hypothesis $H_2 = \{\mu_1<\mu_2=\cdots=\mu_n\}$. The partitioning scheme includes also the hypothesis $J_2=\{\mu_1>\mu_2=\cdots=\mu_n\}$. Consider the union of these two hypotheses, that is $U_2 = H_2\cup J_2 = \{\mu_1\neq\mu_2=\cdots=\mu_n\}$. If we have a valid test for the set $U_2$, then we have a valid test for both $H_2$ and $J_2$. Indeed, since $H_2\subset U_2$, then $LR(H_2)\geq LR(U_2)$ and thus
\begin{equation}
\mathbb{P}_{H_2}(LR(H_2)>\gamma) \leq \mathbb{P}_{H_2}(LR(U_2)>\gamma).
\label{eqn:CoservativeBoundChisq}
\end{equation}
Inside the likelihood ratio LR$(U_2)$, the sum is minimized under $U_2$ inside $U_2$ whatever the relative position of $y_1$ with respect to the weighted average of $y_2,\cdots,y_n$, and the corresponding likelihood ratio equals
\[LR(U_2) = \sum_{i=2}^n{\frac{1}{\sigma_i^2}\left(y_i - \frac{1}{\sum_2^n{\frac{1}{\sigma_j^2}}}\sum_{k=2}^n{y_k/\sigma_k^2}\right)^2}.\]
The probability distribution of $LR(U_2)$ under the hypothesis $H_2$ is a $\chi^2(n-2)$. Hence, in order to test the hypothesis $H_2$ at level $\alpha$, using (\ref{eqn:CoservativeBoundChisq}) it suffices to choose $\gamma$ as the quantile of a $\chi^2(n-2)$ at level $1-\alpha$. Notice that the new critical value is calculated under the null hypothesis and not under the least favorable distribution. The same result applies for the hypothesis $J_2$.\\
This idea is easily generalized to any hypothesis in the partitioning scheme. In other words, in order to test some hypothesis $H$, it suffices to calculate the critical value of a $\chi^2$ with a number of degrees of freedom equal to $n-l-1$ where $l$ corresponds to the number of inequalities in the configuration defining the tested hypothesis $H$.
\end{proof}
The new critical value is now calculated under the null hypothesis and not under the least favorable distribution. Besides, in comparison to the result of the previous paragraph, the new critical value is a quantile of a $\chi^2(n-2)$ instead of the quantile of a mixture of a $\chi^2(n-2)$ and a $\chi^2(n-1)$, which means we got a lower critical value than before and thus a more powerful test. It is very important to notice that the improvement here does not only concern the power of the test. The calculation of the critical level based on the results presented in the previous paragraph is very complicated in the not-equal-$\sigma$ case, because we can no longer use the result of \cite{Miles59}. According to  Gr\"omping (\cite{IcInferPackage}), the calculus takes 9 seconds when $n=10$ and explodes to 1265 seconds for $n=15$. One can imagine what may happen with $n=50$ or higher. Moreover, even in the equal-$\sigma$ case, calculating the Stirling numbers for $n>175$ produces and overflow in the statistical program \texttt{R}, see function \texttt{Stirling1} from package \texttt{copula}. On the other hand, calculating the quantiles of $\chi^2$ distribution is done instantly for any (reasonable) value of $n$.

\subsection{Type I error in terms of the confidence intervals for the ranks: further adjustments on the critical value}\label{subsec:AdjustCritVal}
When we test the partitions, as long as we use any of the critical values of the two previous paragraphs, either the $\chi^2(n-l-1)$ or the mixture of Chi squares, we are protecting against type I error from a point of view of the partitioning itself without taking into account the objective behind all the testing which is the confidence intervals for ranks. This results in more conservative tests. A type I error from a point of view of the partitioning happens when we reject some partition $H$ although the true set of centers $(\mu_1,\cdots,\mu_n)$ is inside it. A type I error related to the ranks happens when we falsely reject not only the partition containing the true set of centers, but all hypotheses which imply larger confidence intervals for the ranks (hypotheses which has at least the same equalities, see Definition \ref{def:InclusionConfigs}). Thus, we should adjust the critical value for the hypotheses so that the protection of type I error is for the ranks and not for the partitions.\\
We start from the top hypothesis with equalities between all the centers. The top hypothesis is the first one to be tested and it includes all other hypotheses since it implies the trivial confidence intervals. Testing the hypothesis with all equalities between the individuals is simply done by comparing the LR with the quantile of a $\chi^2(n-1)$. Now let us take one of the hypotheses from the second level of the partitioning scheme. Consider the hypothesis $\mu_1=\cdots=\mu_{n-1}<\mu_n$. Let $\gamma_{n-1}$ be the quantile of order $1-\alpha$ of the $\chi^2(n-1)$. Let $\gamma>0$ be some real number to be determined later on. We falsely reject the ranks implied by the hypothesis $\mu_1=\cdots=\mu_{n-1}<\mu_n$ if we reject the top hypothesis at the critical level $\gamma_{n-1}$ and the current hypothesis at the critical level $\gamma$. The value of $\gamma$ must ensure that the probability that the two events happen in the same time is at most $\alpha$, that is
\[\mathbb{P}\left(LR(\mu_1=\cdots=\mu_{n-1}<\mu_n) > \gamma ,\; LR(\mu_1=\cdots=\mu_n)>\gamma_{n-1}\right) \leq \alpha.\]
It is interesting to notice that by taking $\gamma = \gamma_{n-2}$, we have
\begin{multline*}
\mathbb{P}\left(LR(\mu_1=\cdots=\mu_{n-1}<\mu_n) > \gamma_{n-2} ,\; LR(\mu_1=\cdots=\mu_n)>\gamma_{n-1}\right) \leq \\ \mathbb{P}\left(LR(\mu_1=\cdots=\mu_{n-1}<\mu_n) > \gamma_{n-2}\right) \leq \alpha.
\end{multline*}
Thus, the proposed adjustment from paragraph \ref{subsec:AldoAdjust} is already included in the new definition for the critical value. Moreover, it shows clearly that by following this lead, it is possible to find a lower critical value than $\gamma_{n-2}$.\\
The logic is to decompose this probability into two cases according to the relative position of $y_n$ with respect to the weighted average of the observations $y_1,\cdots,y_{n-1}$. 
\begin{multline*}
\mathbb{P}\left(LR(\mu_1=\cdots=\mu_{n-1}<\mu_n) > \gamma ,\; LR(\mu_1=\cdots=\mu_n)>\gamma_{n-1}\right) = \\ \mathbb{P}\left(LR(\mu_1=\cdots=\mu_{n-1}<\mu_n) > \gamma ,\; LR(\mu_1=\cdots=\mu_n)>\gamma_{n-1}\left| y_n > \frac{1}{\sum{\frac{1}{\sigma_j^2}}}\sum_{k=1}^{n-1}{ y_k/\sigma_k^2}\right.\right)\times\\
\mathbb{P}\left(y_n > \frac{1}{\sum{\frac{1}{\sigma_j^2}}}\sum_{k=1}^{n-1}{ y_k/\sigma_k^2}\right) + \\
 \mathbb{P}\left(LR(\mu_1=\cdots=\mu_{n-1}<\mu_n) > \gamma ,\; LR(\mu_1=\cdots=\mu_n)>\gamma_{n-1}\left| y_n \leq \frac{1}{\sum{\frac{1}{\sigma_j^2}}}\sum_{k=1}^{n-1}{ y_k/\sigma_k^2}\right.\right) \times \\ \mathbb{P}\left(y_n \leq \frac{1}{\sum{\frac{1}{\sigma_j^2}}}\sum_{k=1}^{n-1}{ y_k/\sigma_k^2}\right).
\end{multline*}
Notice that when $y_n \leq \frac{1}{\sum{\frac{1}{\sigma_j^2}}}\sum_{k=1}^{n-1}{ y_k/\sigma_k^2}$, then the minimum in the likelihood ratio under the null hypothesis will be attained on the border of the the set $\{\mu_1=\cdots=\mu_{n-1}<\mu_n\}$. Thus $LR(\mu_1=\cdots=\mu_{n-1}<\mu_n) = LR(\mu_1=\mu_2=\cdots=\mu_n)$. This implies that if $\gamma\leq\gamma_{n-1}$, the second term in the previous display simplifies. 
\begin{multline*}
\mathbb{P}\left(LR(\mu_1=\cdots=\mu_{n-1}<\mu_n) > \gamma ,\; LR(\mu_1=\cdots=\mu_n)>\gamma_{n-1}\right) = \\ \mathbb{P}\left(\sum_{i=1}^{n-1}{\frac{1}{\sigma_i^2}\left(y_i - \frac{1}{\sum{\frac{1}{\sigma_j^2}}}\sum_{k=1}^{n-1}{ y_k/\sigma_k^2}\right)^2} > \gamma ,\;\sum_{i=1}^{n}{\frac{1}{\sigma_i^2}\left(y_i - \frac{1}{\sum{\frac{1}{\sigma_j^2}}}\sum_{k=1}^{n-1}{ y_k/\sigma_k^2}\right)^2}>\gamma_{n-1}\right)\\ 
\times \mathbb{P}\left(y_n > \frac{1}{\sum{\frac{1}{\sigma_j^2}}}\sum_{k=1}^{n-1}{ y_k/\sigma_k^2}\right) \\
+ \mathbb{P}\left(\sum_{i=1}^{n}{\frac{1}{\sigma_i^2}\left(y_i - \frac{1}{\sum{\frac{1}{\sigma_j^2}}}\sum_{k=1}^{n-1}{ y_k/\sigma_k^2}\right)^2}>\gamma_{n-1}\right)
\mathbb{P}\left(y_n \leq \frac{1}{\sum{\frac{1}{\sigma_j^2}}}\sum_{k=1}^{n-1}{ y_k/\sigma_k^2}\right)
\end{multline*}
Although the second term contains the probability of rejecting the top hypothesis at level $\alpha$, it is calculated under the hypothesis $\mu_1=\cdots=\mu_{n-1}<\mu_n$, and hence it does not evaluate to $\alpha$. Since the random variable $\sum_{i=1}^{n}{\frac{1}{\sigma_i^2}\left(y_i - \frac{1}{\sum{\sigma_j^2}}\sum_{k=1}^{n-1}{ y_k/\sigma_k^2}\right)^2}$ has a non-central $\chi^2$ distribution under the null hypothesis $\mu_1=\cdots=\mu_{n-1}<\mu_n$ with unknown centrality parameter which depends on the true position of the centers from each others, it \emph{not} is impossible to evaluate the second term. It is however possible to upper-bound it by the corresponding worst case, that is the maximum over all the possible values for the centers inside the null hypothesis $\mu_1=\cdots=\mu_{n-1}<\mu_n$. \\
We will explore this idea in the case of equal standard deviations of the centers in order to keep formulas clearer. We need at first the following lemma.
\begin{lemma}[Theorem 5, \cite{Apostol}] \label{lemm:SumOfSquares}
Let $y_1,\cdots,y_n$ and $\sigma_1,\cdots,\sigma_n$ be real numbers. Let $n_1\in\{1,\cdots,n\}$. 
\begin{multline*}
\sum_{i=1}^n{\left(y_i - \frac{1}{\sum_{j=1}^n{\sigma_j^2}}\sum_{k=1}^n{\sigma_k^2y_k}\right)^2} = \sum_{i=1}^{n_1}{\left(y_i - \frac{1}{\sum_{j=1}^{n_1}{\sigma_j^2}}\sum_{k=1}^{n_1}{\sigma_k^2y_k}\right)^2} + \sum_{i=n_1+1}^{n}{\left(y_i - \frac{1}{\sum_{j = n_1+1}^n{\sigma_j^2}}\sum_{k=n_1+1}^{n}{\sigma_k^2y_k}\right)^2} \\
+ \frac{\sum_{j=1}^{n_1}{\sigma_j^2}\sum_{j=n_1+1}^{n}{\sigma_j^2}}{\sum_{j=1}^{n_1}{\sigma_j^2}+\sum_{j=n_1+1}^{n}{\sigma_j^2}}\left(\frac{1}{\sum_{j=1}^{n_1}{\sigma_j^2}}\sum_{k=1}^{n_1}{\sigma_k^2y_k}-\frac{1}{\sum_{j=n_1+1}^n{\sigma_j^2}}\sum_{k=n_1+1}^{n}{\sigma_k^2y_k}\right)^2.
\end{multline*}
\end{lemma}
Using this lemma and setting $\sigma_i=\sigma$, we may write
\[\frac{1}{\sigma^2}\sum_{i=1}^n{\left(y_i-\frac{1}{n}\sum{y_i}\right)^2} = \frac{1}{\sigma^2}\sum_{i = 1}^{n-1}{(y_i - \frac{1}{n-1}\sum_{i=1}^{n-1}{y_i})^2} + \frac{1}{\sigma^2}\frac{n-1}{n}\left(y_n - \frac{1}{n-1}\sum_{i=1}^{n-1}{y_i}\right)^2.\]
Now, the non centrality is separated from the main sum of squares. The first term in the RHS has a distribution of $\chi^2(n-2)$ under the null hypothesis. The second term in the RHS is a $\chi^2_{(\mu_n-\mu_1)^2}(1)$, that is a non-central $\chi^2$ distribution with 1 degrees of freedom and a non-centrality parameters equal to $(\mu_n-\mu_1)^2$. Moreover, these two $\chi^2$ random variables are independent from each other due to Cochran's theorem and to the fact that $y_n$ is independent from the other observations and does not appear in either $\frac{1}{n-1}\sum_{i=1}^{n-1}{y_i}$ or in the first term in the RHS. We now have
\begin{eqnarray*}
\mathbb{P}\left(\frac{1}{\sigma^2}\sum_{i=1}^n{\left(y_i-\frac{1}{n}\sum{y_i}\right)^2}>\gamma_{n-1}\right) & = &
\mathbb{P}\left(\frac{1}{\sigma^2}\sum_{i = 1}^{n-1}{\left(y_i - \frac{1}{n-1}\sum_{i=1}^{n-1}{y_i}\right)^2}\right. \\
& & \qquad \left.+ \frac{1}{\sigma^2}\frac{n-1}{n}\left(y_n - \frac{1}{n-1}\sum_{i=1}^{n-1}{y_i}\right)^2>\gamma_{n-1}\right) \\ 
& = & \int_{x+y>\gamma_{n-1}}{f_{\chi^2(n-2)}(x)f_{\chi^2_{(\mu_n-\mu_1)^2}(1)}(y)dydx} \\
& = & \int_{x=0}^{\infty}{f_{\chi^2(n-2)}(x)\left[1-\mathbb{F}_{\chi^2_{(\mu_n-\mu_1)^2}(1)}(\gamma_{n-1}-x)\right]dx}.
\end{eqnarray*}
Similarly, we have
\begin{multline*}
\mathbb{P}\left(\frac{1}{\sigma^2}\sum_{i=1}^{n-1}{\left(y_i-\frac{1}{n-1}\sum{y_i}\right)^2}>\gamma,\frac{1}{\sigma^2}\sum_{i=1}^n{\left(y_i-\frac{1}{n}\sum{y_i}\right)^2}>\gamma_{n-1}\right)  = \\ \int_{x=\gamma}^{\infty}{f_{\chi^2(n-2)}(x)\left[1-\mathbb{F}_{\chi^2_{(\mu_n-\mu_1)^2}(1)}(\gamma_{n-1}-x)\right]dx}.
\end{multline*}
Recall that we are looking for $\gamma$ such that
\begin{multline}
\frac{1}{2}\mathbb{P}\left(\frac{1}{\sigma^2}\sum_{i=1}^{n-1}{\left(y_i-\frac{1}{n-1}\sum{y_i}\right)^2}>\gamma,\frac{1}{\sigma^2}\sum_{i=1}^n{\left(y_i-\frac{1}{n}\sum{y_i}\right)^2}>\gamma_{n-1}\right) \\ + \frac{1}{2}\mathbb{P}\left(\frac{1}{\sigma^2}\sum_{i=1}^n{\left(y_i-\frac{1}{n}\sum{y_i}\right)^2}>\gamma_{n-1}\right) \leq \alpha.
\label{eqn:AdjustCritValEqSig}
\end{multline}
The sum of the two probabilities increases as the difference $\mu_n - \mu_1$ increases. Notice that under the null hypothesis $\mu_1=\cdots=\mu_{n-1}<\mu_n$, if we suppose that $\mu_n - \mu_1 = 0.587$ and for $n = 10$, then by maximizing over the difference $\mu_n-\mu_1$ inside $[0,0.587]$, the value of $\alpha$ that results in an equality in (\ref{eqn:AdjustCritValEqSig}) is approximately the quantile of $\chi^2(n-2)$. This corresponds to our first adjustment given in paragraph \ref{subsec:AldoAdjust}. If we can guarantee that $\mu_n - \mu_1 < 0.587$, then an adjustment to the critical value of the quantile of $\chi^2(n-2)$ can be made. The lower the upper bound on the difference $\mu_n - \mu_1$ can be ensured, the better the adjustment on the critical value can be made. A similar but more complicated treatment can be done on hypotheses from lower levels (with higher number of inequalities) to deduce another upper bound on the differences between the centers. This adjustment appears to be relevant only if the centers are close enough from each others. Otherwise, no adjustment can be made by maximizing over the difference $\mu_n-\mu_1$ simply because the value of the second probability in the LHS of (\ref{eqn:AdjustCritValEqSig}) starts to exceed $2\alpha$.

\section{Some shortcuts and practical issues}\label{sec:shortcuts}
The main idea in the partitioning principle is to generate all possible simple disjoint hypotheses such that null hypothesese of interest are included inside unions of some of them. The number of partitions is generally exponentially related to the number of observations and depends of course on how fine is the proposed partitioning. In order to illustrate the complex problem we are treating and its difficulty, the following lemma gives us an idea about the number of hypotheses to be tested so that we cover all the partitions described in Section \ref{sec:HowToRank}.
\begin{lemma} 
The number of correctly ordered hypothesis w.r.t. the empirical ordering induced by the observations is equal to $2^{n-1}$. Moreover, the total number of hypotheses in the partitioning scheme is 
\begin{equation}
\sum_{l=2}^{n} C_n^l (n-l)! + n! = \sum_{l=1}^{n}{\frac{n!}{l!}}
\label{eqn:NbConfigPartition}
\end{equation}
where $0!=1$. 
\end{lemma}
\begin{proof}
Since we have for each relation between two consecutive individuals only two options $=$ or $<$, then the overall number of correctly ordered hypotheses is $2^{n-1}$.\\
In order to do the calculus for the whole partitioning, we calculate the number of hypotheses in each level and then sum them up, that is the hypotheses having the same number of inequalities between the centers. Start with the bottom level with only inequalities and then go up level by level by inserting equalities. 
\begin{enumerate}
\item The bottom level with only inequalities: We consider that the individuals are ordered in each hypothesis in an ascending order so that we only need to think about the number of permutations of the individuals in order to cover all cases. Thus the number of hypotheses at the bottom level is equal to the number of permutations in the set $\{1,\cdots,n\}$ with one orbit. This is equal to $n!$.
\item Level $n-1$ with only one equality: We should consider the number of possibilities of equality between two individuals and then multiply it by the number of configurations of inequalities for each equality position. Repetitions are not accepted. The number of possibilities for equality between the first individual and another one from the list is equal to $n-1$. The remaining number of possibilities for the second one is $n-2$. Thus, the whole number of possibilities without repetitions is equal to $\sum_{k=1}^{n-1}{(n-k)}$. It is the same as counting the number of subsets of length 2 from the whole set of individuals, and this is equal to $C_n^2=n(n-1)/2$. Now that two elements are merged into one with an equality, we need to count down all permutations in a set of $n-1$ elements as in the previous point. This gives $(n-1)!$ configurations. Hence, the number of configurations in level $n-1$ is equal to
\[(n-1)!\sum_{k=1}^{n-1}{(n-k)} = \frac{(n-1)n!}{2}.\]
\item Level $n-2$ with two equalities. This time we need to choose for two equalities. We are then counting the number of subsets of length 3 inside $\{1,\cdots,n\}$ which yields $C_n^3$. Now that two equalities are constructed, the remaining number of elements is $n-2$ and the number of permutations (in order to construct the inequalities) is equal to $(n-2)!$. Hence, the number of configurations in level $n-2$ is equal to
\[(n-2)!C_{n}^3\]
\item For level $n-3$, a similar argument shows that the overall number of configurations is
\[(n-3)!C_{n}^4\]
\end{enumerate}
This builds up the logic for finding the whole number of configurations in the whole partitioning scheme given by formula \ref{eqn:NbConfigPartition}.
\end{proof}
This lemma shows that the complexity of the partitioning scheme is super exponential in the number of observations. It is thus important to find out a way to avoid this number as much as possible and preferably reduce the complexity into polynomial. We give a useful notion that will be used later in the sequel.\\
\begin{definition}[Significance of a hypothesis]
Suppose that the partitioning scheme is tested in the following order $H_1, \cdots, H_t$. We say that a hypothesis $H_i$ for $i\in\{1,\cdots,t\}$ is significant in the partitioning scheme if it implies an increment in the ranks induced from testing the hypotheses $H_1,\cdots,H_{i-1}$.
\end{definition}
This definition suggests that in the algorithm built to test the partitions, we may at each step check if the current hypothesis is significant to the ranking induced from the previously tested hypotheses. If it is not significant, then there is no need to test it, because if it is rejected then it does not help us decide anything about the ranks; otherwise (if it is not rejected) this hypothesis will produce smaller confidence intervals than what we already have or at the best the same confidence intervals. This idea can be very efficient if testing a hypothesis costs too much execution time especially if it requires doing an optimization over some parameters. In the case of the LR test and a Gaussian model, the improvement from checking the significance of a hypothesis was minor.\\
We move now to more interesting ideas which can help reduce the super exponential complexity of the partitioning. In order to find a shortcut inside the partitioning scheme, we need to suppose first that the observed values are ordered, that is $y_1<y_2<...<y_n$. We start by the most simple and easy shortcut which is based on the idea of the Partitioning principle that in order to reject a hypothesis, all sub-hypotheses or partitions must be rejected. 
\begin{lemma}
If the top hypothesis is not rejected, then there is no need to test any other hypothesis because the confidence intervals for the ranks are the trivial ones, that is $[1,n]$.
\end{lemma}
Now comes the most important result in this section.
\begin{proposition}\label{lemm:ShortcutCorrectHyp}
It suffices to test the hypotheses whose ordering coincides with the empirical ordering that is the one which is induced by the observed data.
\end{proposition}
\begin{proof}
Consider a hypothesis $H_l$ from the $l^{\text{th}}$ level. This hypothesis contains $l-1$ inequality. Write this hypothesis as a union of blocks where each block contains all elements related with each other by an equality, that is $H_l = \{A_1,\cdots,A_l\}$. \\
Suppose that this hypothesis does not follow the empirical ordering. If $H_l$ is rejected, then it adds nothing to the confidence intervals of the ranks. Suppose then that the hypothesis $H_l$ is not rejected. When we calculate the maximum likelihood under this hypothesis by the PAVA, adjacent blocks which do not respect the empirical ranks will be pooled together. By merging the pooled blocks of hypothesis $H_l$, we can construct a well ordered hypothesis $\bar{H}_{s} = \{\tilde{A}_{1},\cdots,\tilde{A}_{s}\}$ with $s<l$ whose empirical (weighted) averages inside each one of its blocks follow the order imposed by the hypothesis itself. This hypothesis $\bar{H}_{s}$ clearly has the same LR as the original one $H_l$. Since $\bar{H}_s$ comes from an upper level (level number $s$), the associated critical level is higher. Thus, the non rejection of $H_l$ will imply the non rejection of the hypothesis $\bar{H}_s$. Since both hypotheses induce the same confidence intervals for the ranks, then we can consider only one of them which will be the correctly ordered one. \\
\end{proof}
This lemma is interesting in itself because of its generality. The model does not interfere in the proof and the idea can be adapted to other types of test statistics other than the LR test. Besides, the critical value does not interfere either as long as it is non increasing along the levels. On the other hand, using this result we were able to reduce the complexity of the partitioning procedure from $\sum_{l=1}^{n}{n!/l!}$ to $2^{n-1}$. One can imagine the gain by looking at Stirling's approximation. If we only consider the first term in the sum, that is $n!$. We have $n!\approx \sqrt{2\pi n} \left(\frac{n}{e}\right)^n$. Thus, we reduced the number of tested hypotheses by a factor of more than $2\sqrt{2\pi n} \left(\frac{n}{2e}\right)^n$. For example, for $n=10$, the total number of partitions is 6235301 whereas it becomes 512 using Proposition \ref{lemm:ShortcutCorrectHyp}.\\
It is important to point out that since we are only considering the hypotheses for which the empirical ranking does not contradict with the ordering imposed by the hypotheses themselves. Therefore, the maximum of the likelihood under the null hypothesis is always attained inside the null and there is no need to use the PAVA in order to calculate the maximum likelihood. The maximum of the likelihood under the null is obtained simply by averaging inside each group of equalities between the centers.\\

We provide now two stopping criterion for an algorithm which tests the partitions level by level. They can be useful in some cases especially if the confidence intervals are very wide.
\begin{enumerate}
\item If we do not reject any hypothesis at some level $l$, then there is no need to check any hypothesis from lower levels (with fewer number of equalities) since they would produce at most the same confidence intervals for the ranks.
\item If the current minimum length of both the upper part and the lower part of the confidence interval of center $\mu_i$ is equal to the number of equalities in the current level, then there is no longer any need to continue further in the partitioning for the center $\mu_i$. Thus if this is fulfilled for all individuals then the procedure should stop and no further information is gained. This is a very good shortcut if the estimated confidence intervals are wide enough. More formally, the greatest confidence interval that the partitions at a level $l$ provide for center $\mu_i$ (with empirical rank equal to $i$) is given by:
\[\left[\max(i-(n-l),1),\; \min(i+n-l,n)\right]\]
If \emph{for all centers} the current confidence interval is larger than this one as we arrive at level $l$, then the testing procedure may stop.
\end{enumerate}

\section{Algorithms to calculate the confidence intervals for the ranks based on the partitioning principle}\label{sec:Algos}
We present two algorithms which produce the confidence intervals for the ranks. The first one uses the partitioning scheme we have considered till now that is by considering $n$ levels where each level $l$ contains all hypotheses with $l-1$ inequality for $l=1,\cdots,n$. The second algorithm focuses on the blocks of centers related by equalities in the hypotheses starting at some center. Each algorithm has its own properties and tips. Although in practice, the second one proved to be more efficient, we present the two of them in case new tips and shortcuts are found by an interested reader for the first one. Note also that only correctly ordered hypotheses are considered due to Proposition \ref{lemm:ShortcutCorrectHyp}. \\
A key idea is how to represent or code the hypotheses on the one hand, and then how to generate these codes efficiently. For each algorithm, we will present a way to generate the list of hypotheses efficiently. Other possibilities can exist and finding a simpler way to generate and keep track of the hypotheses may improve significantly the performance of the algorithms. 
\subsection{A level-by-level algorithm}
In this algorithm, the idea is to use the partitioning scheme presented in figure (\ref{fig:Partitions3by3}) for three centers. We take into account that hypotheses whose ordering do not comply with the empirical ranking are not needed due to Proposition \ref{lemm:ShortcutCorrectHyp}. In practice, it is not possible to code all the hypotheses prior to the testing, because it concerns keeping in hand a matrix of size $2^{n-1}\times c$ where $c$ is the length (or the lengths) of the representation. Thus, for "normal" computers it becomes easily impossible to generate such matrix (or structure) as $n$ grows. Therefore, it is necessary to be able to generate the configurations (representations) one by one to avoid memory problems. \\
We propose to represent a hypothesis by keeping track of the positions of the inequalities so that a hypothesis is cut according to the blocks of equalities between the centers. It allows an efficient way when we want to calculate the LR. Indeed, since we only test hypotheses with a correct ordering w.r.t the empirical one, the PAVA is not needed and the LR for some partition is a mere sum of averages of the blocks of equal centers and our representation tells us directly where are the bounds of each block. For example, in the case of 3 centers, the representation of the correctly ordered hypotheses (without the top hypothesis) is the set
\[\{(0),(1),(0,1)\}.\]
The top hypothesis with no inequalities is generally tested separately and thus not included in the representation.
 For reasonable values of $n$, it is possible to use function \texttt{combn} from the \texttt{utils} package in the statistical program \texttt{R} in order to generate efficiently the set of configurations level by level\footnote{In our case, it worked for $n<26$ with a laptop provided with 8 GB of RAM.}. It is unfortunately not possible to generate these configurations one by one using the method programmed inside function \texttt{combn} so that we avoid generating the whole matrix at once. We propose a different way to do so by considering the following function. Let $(c_1,\cdots,c_k)$ with $c_i,\cdots,c_k$ be non negative integers such that $c_1<\cdots<c_k$. Define function $C$ as follows
\[C(c_1,\cdots,c_k) = \binom{c_1}{1} + \cdots + \binom{c_k}{k}.\]
where $\binom{a}{b}$ is the binomial coefficient. Function $C$ is known as the combinatorial number system in combinatorics \cite{Knuth}. This is a one-to-one function between the set of configurations 
\[\{(c_1,\cdots,c_k)\in\mathbb{N}^k, 0\leq c_1<\cdots<c_k\leq n-2,\}\]
and the set of numbers $\{1,\cdots,\binom{n-1}{k}\}$. For our purpose, we use the reciprocal function. At each level $k$, we iterate on the numbers between $1$ and $\binom{n-1}{k-1}$ (recall that the first level is the top one with no inequalities so that $k$ starts at 2), and then calculate for each number its inverse using function $C^{-1}$. This is done iteratively as follows. For a number $i$, we look for the maximum number $c_k$ such that $\binom{c_k}{k}\leq i$. Then, we subtract $\binom{c_k}{k}$ from $i$ and look for the maximum number $c_{k-1}$ such that $\binom{c_{k-1}}{k-1}\leq i - \binom{c_k}{k}$ and so on until we find $c_1$. This way we produce the set of correctly ordered hypotheses in level $k$ when $i$ spans the whole set $\{1,\cdots,\binom{n-1}{k}\}$.\\
\begin{algorithm}[H]
\label{algo:LevelLevel}
\KwData{Ordered sample $y_1,\cdots,y_n$ and standard deviations $\sigma_1,\cdots,\sigma_n$. Confidence level $\alpha$.}
\KwResult{For each $i, [a_i,b_i]$ such that $\mu_i\in[a_i,b_i]$ with joint probability greater than $1-\alpha$.}
\eIf {the top hypothesis is not rejected}
{Set confidence intervals to $[1,n]$; $\forall i, a_i=1, b_i = n$.}
{
 \For{$l$ from $2$ to $n-1$}
 {
		m = choose($n-1,l-1$) \;
		Generate all hypotheses with $l-1$ inequalities $H_{1,l},\cdots,H_{m,l}$\;
		\For{$i$ from 1 to $m$}
		{
			\If{$H_{i,l}$ is significant to the actual ranking}
			{
				Calculate the LR under $H_{i,l}$ \;
				\If{LR$(H_{i,l})<\chi_{1-\alpha}^2(n-l)$}
				{
					Update the ranks.
				}
			}
		}
		\If{$\forall i, \left[\max(i-(n-l),1),\; \min(i+n-l,n)\right]\subset [a_i,b_i]$}
		{
			Break the loop over $l$.
		}
 }
}
\caption{A level-by-level algorithm}
\end{algorithm}
In Algorithm \ref{algo:LevelLevel}, we avoided the details of generating the representation of the hypotheses using the combinatorial number system in order to keep the algorithm clear and easily readable. We attract the attention of the interested that since we are calculating the LR for each tested hypothesis that is $2^{n-1}$, there is no need to recalculate each time the averages inside each block with equalities between the centers. We calculate the quantities $\{LR(\mu_i=\cdots=\mu_j),1\leq i< j\leq n\}$ separately prior to the testing procedure, so that we only sum up the parts of the log-likelihood according to the configuration defining the partition. Moreover, if we are using the combinatorial number system, the decoding of the numbers requires the calculus of too many binomial coefficients, and these need not be recalculated each time needed and can be calculated once and for all prior to the testing procedure, that is the set $\{\binom{j}{i},1\leq i< j\leq n\}$. These two remarks are the most important tips when programing the partitioning procedure in Algorithm \ref{algo:LevelLevel} and help reduce the execution time greatly.\\
There is still a room for improvement, but it will break down the level-structure. We can code the hypotheses using binary numbers by setting for example 0 for equality and 1 for inequality (see \cite{Knuth} for more insights). Thus, hypotheses from the 3-centers example are coded by $\{(0,0),(1,0),(0,1),(1,1)\}$ (the top hypothesis included). Generating these configurations can be done by converting numbers between 0 and $2^{n-1}-1$ into binary. Thus, we no longer take into account which level we are considering. We found out that the best way to make use of this binary coding is by saving a sparse representation of it using again decimal numbers. As we calculate the residual of dividing the number by 2, we save only the position of the one (or the zeros) whenever it appears. It produces a different way to represent the hypotheses by keeping track of inequality positions but without respecting any leveling structure of the partitioning. This might be in some situations not the best way to do it especially if we are able to find more efficient shortcuts using the level-structure other than what we found here. In our situation and with only the shortcut proved in Proposition \ref{lemm:ShortcutCorrectHyp}, the level-structure programmed using the combinatorial number system can be very efficient if we are expecting either very narrow or very wide confidence intervals. Otherwise, the binary structure is far more efficient.
\subsection{A block-check algorithm}\label{subsec:BlocksApproach}
The idea behind this algorithm is based on how we construct a confidence interval for the rank of some center. The easiest way to explain it is to take the case of the center with the lowest empirical rank. The confidence interval for its rank has the form $[1,a_1]$. This confidence interval is obtained from not rejecting any hypothesis of the form $\{\mu_1=\cdots=\mu_{a_1}<\cdots\}$. In order to find one of these hypotheses, we start with the largest block starting at the first center. This may be $\{\mu_1=\cdots=\mu_n\}$ or any smart guess (we will see later an efficient way to do so). We then move on to a smaller one and always by testing all possible configurations which can be added at the end. In other words, assume we want to test the block $\{\mu_1=\cdots=\mu_r\}$. We need to check all hypotheses of the form $\{\mu_1=\cdots=\mu_r<\cdots\}$. It is as if we are doing a sub-partitioning on the set of centers $\{\mu_{r+1},\cdots,\mu_n\}$ and then concatenating these partitions with the block of interest $\{\mu_1=\cdots=\mu_r\}$ when we calculate the LR statistic. Whenever one of these hypotheses ($\{\mu_1=\cdots=\mu_r\}$+sub-partition) is not rejected, we say that the block of interest $\{\mu_1=\cdots=\mu_r\}$ is not rejected. Moreover, there is no need to continue further with center $\mu_1$ because this is sufficient to deduce that $a_1=r$ according to the partitioning principle, and the confidence interval for the rank of $\mu_1$ is $[1,r]=[1,a_1]$. For centers in the middle, say center $\mu_i$, we need to count for two sub-partitionings; one for the left part that is $\{\mu_1,\cdots,\mu_{i-1}\}$ and one for the right part that is what is left from the end of the block till the highest center $\mu_n$. In order to generate the left and right sub-partitions, we use for example a binary-decimal representation as in Algorithm 1.\\
Figure (\ref{fig:BlockBasedPartitioning}) shows the blocks that we need to test and the complementary subset of centers. For a block of size 3 starting at center $A$, the sub-partitions that we have to the right contains but one center here, that is center $D$. For a block of size 2 starting at center $A$, the sub-partitions to the right contains 2 centers, and thus two sub-partitions $\{C=D\}$ and $\{C<D\}$ need to be considered, and are then concatenated to the original block to form the 3rd and 4th hypotheses in the column of center $A$. Last but not least, for a block of size 1 containing only center $A$, the sub-partitioning to the right contains 4 cases as shown in the scheme. These hypotheses need not be tested, because we reached the minimum block size for a center which is 1. We kept them however to illustrate the idea of the algorithm. Generally, a smart guess for a minimum-block size can be calculated easily so that we stop at an earlier stage.\\
\begin{figure}[ht]
\centering
\includegraphics[scale=0.6]{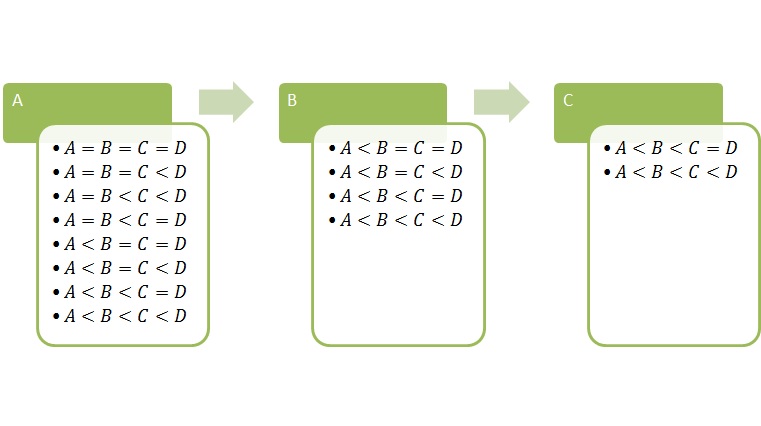}
\caption{All possible sub-hypotheses that need to be tested in a block-based algorithm. In this example, we consider 4 centers. For the 4th center, there is no hypothesis to test.}
\label{fig:BlockBasedPartitioning}
\end{figure}

It is true that for each middle block, we need at worst case scenario to perform two sub-partitions (left and right complements of the block). However, the complexity will be smaller than the complexity of treating the whole set of centers. For example, if we are treating a block containing $k$ centers such that there remains $b$ centers to the left and $d$ centers to the right. We need at worst case scenario to perform $2^{b+d-2}$ tests in order to confirm the rejection or the non rejection of this block. Thus, depending on the prior information we have on the blocks, we may improve upon the execution time. The complexity of the algorithm can be estimated to be $\mathcal{O}(2^{n-c})$ where $c$ is the the minimum of estimated minima block sizes for the centers\footnote{This minimum is generally different from min(\texttt{MinBlock}) shown in Algorithm \ref{algo:BlockPartitioning}. Indeed, the largest center will always have a minimum block size of 1. Therefore, it is important to look for the first center whose rank confidence interval contains the last center, say $\mu_{n_0}$. The minimum block size is then calculated only for centers $\mu_1\cdots,\mu_{n_0}$.}. Worst case scenario, we perform as much as we do in Algorithm \ref{algo:LevelLevel} which can happen if we have a center with a point-wise confidence interval or without relying on any prior information that is $c=1$. On the other hand, the best case scenario happens if we are lucky enough. This happens when we have a good prior information about the maximum and the minimum size for the blocks (see \ref{rem:BestMinMaxBlockSize}) and that we find from the beginning a hypothesis for which the maximum-sized block is not rejected. In this case, we won't need to perform the whole $2^{b+d-2}$ tests, but only a small portion of them. This is too optimistic, but if we combine this algorithm with the ideas of the next section, it might happen sometimes. Indeed, the next section provides a way to find a minimum and a maximum size for the blocks necessary for building the confidence interval for each of the centers. The result is most of the time very accurate that the difference between the minimum and the maximum is only 1, so that we would only be checking whether it is the minimum or the maximum which equals to the border of the confidence interval. \\
It is also interesting that we can get a lower bound on the size of the blocks by testing a subset of the partitions. For example, we can test all the blocks without any additional equality in the left or in the right sub-partitions. In Figure (\ref{fig:BlockBasedPartitioning}), this corresponds to testing $\{A=B=C=D\}$, $\{A=B=C<D\}$, $\{A=B<C<D\}$, $\{A<B=C=D\}$, $\{A<B=C<D\}$ and $\{A<B<C=D\}$. This can be done with a quadratic complexity $\mathcal{O}(n(n-1))$. This gives in practice a very close lower bound for the true size of the blocks and can be considered as a starting point. Notice that we can also use any low-complexity method which produces confidence intervals in order to produce the starting point for the confidence intervals as long as we are sure that it provides anti-conservative confidence intervals for the ranks.\\
In Algorithm \ref{algo:BlockPartitioning} in the appendix, we give briefly a pseudo-code for our approach based on testing the blocks. The \texttt{MinBlock} and \texttt{MaxBlock} are vectors containing minimum and maximum blocks sizes for each center. It is important to use the confidence intervals implying the minimum block sizes so that we build upon them the exact ones. So, if we calculate the minimum block size through testing blocks hypotheses, the resulting confidence intervals must be used as initial intervals for the algorithm. If we use the result of the next section, we must use the lower-bound confidence intervals as initial intervals for the algorithm and update them whenever a new hypothesis is not rejected.\\

\section{Bracketing: a polynomially-complex algorithm}\label{sec:ApproxPartition}
The challenge in the algorithms we presented earlier was to provide a fast way to go through the partitions. Still; the exponential number of tests required makes it impossible to obtain a fast algorithm for relatively large samples, and any improvement on the algorithm itself or the code used to compile it (with simple or advance programming languages) will easily disappear as $n$ grows. However, if we are willing to loosen up and trade complexity with scalability, it is possible to provide approximate results for the confidence intervals. In this section, we are going to provide a very narrow bracketing for the "exact-partitioning" confidence intervals using a polynomially-complex algorithm which generally in practice has a gap not exceeding one rank per center. The idea of the algorithm is based on two things. the first is to make use of the almost linearity of the quantile of the $\chi^2$ distribution when the number of degrees of freedom becomes high enough. It is natural to assume this since high degrees of freedom appear when we have a great number of centers. The case of small number of degrees of freedom appear either if we have a small number of centers (where Algorithms \ref{algo:LevelLevel}, \ref{algo:BlockPartitioning} can be used) or in special situations where we have very distant centers. The second idea of the algorithm is to make use of the form of the test statistic which is a sum of sums where each one of them corresponds to the share of a block of equal centers (under the null). The gain from this approximation is mainly two things; obtain a good approximation of the result of an exact partitioning for large sample sizes where non of our proposed algorithms from the previous section is able to deliver a result within a reasonable time, and to a provide good upper and a lower bounds for the block sizes of the centers so that they can be used in Algorithm \ref{algo:BlockPartitioning} to obtain an exact-partitioning result for moderate sample sizes (less than 50 in our case).\\
Figure (\ref{fig:Chi2Approx}) shows two linearizations for the qunatile of the $\chi^2$; an upper bound and a lower bound. If the quantile of the $\chi^2$ was a linear function, then the LR statistic can be regarded as a sum of contributions. Each equality between two centers under the null adds a new term to the LR statistic and implies on the other hand an addition of 1 degrees of freedom to the $\chi^2$ quantile against which we are testing the null. In the same spirit as paragraph \ref{subsec:BlocksApproach}, we test a block by adding (at most) two sub-partitions; one to the left of the block and one to the right of it. Each sub-partition contains a subset of the centers, and depending on how these centers are related through the partition (i.e. under the null), their contribution to the LR statistic and to the critical value (the $\chi^2$'s quantile) changes. The linearization of the $\chi^2$ permits to interpret these two additions of sub-partitions as a sum of two contributions to the contribution of the block of interest, and by finding the minimum contribution of these two additions in the whole LR statistic, we can judge to reject or not reject the block of interest. Indeed, if the minimum contribution of both the left and the right additions permits to reject the whole partition, then \emph{any} addition to the block of interest will result in a rejected partition. Otherwise, it will mean that there exists at least one partition containing the block of interest which is not rejected and this is sufficient to include the block of interest in the calculus of the confidence intervals.
\begin{figure}[ht]
\centering
\includegraphics[scale = 0.6]{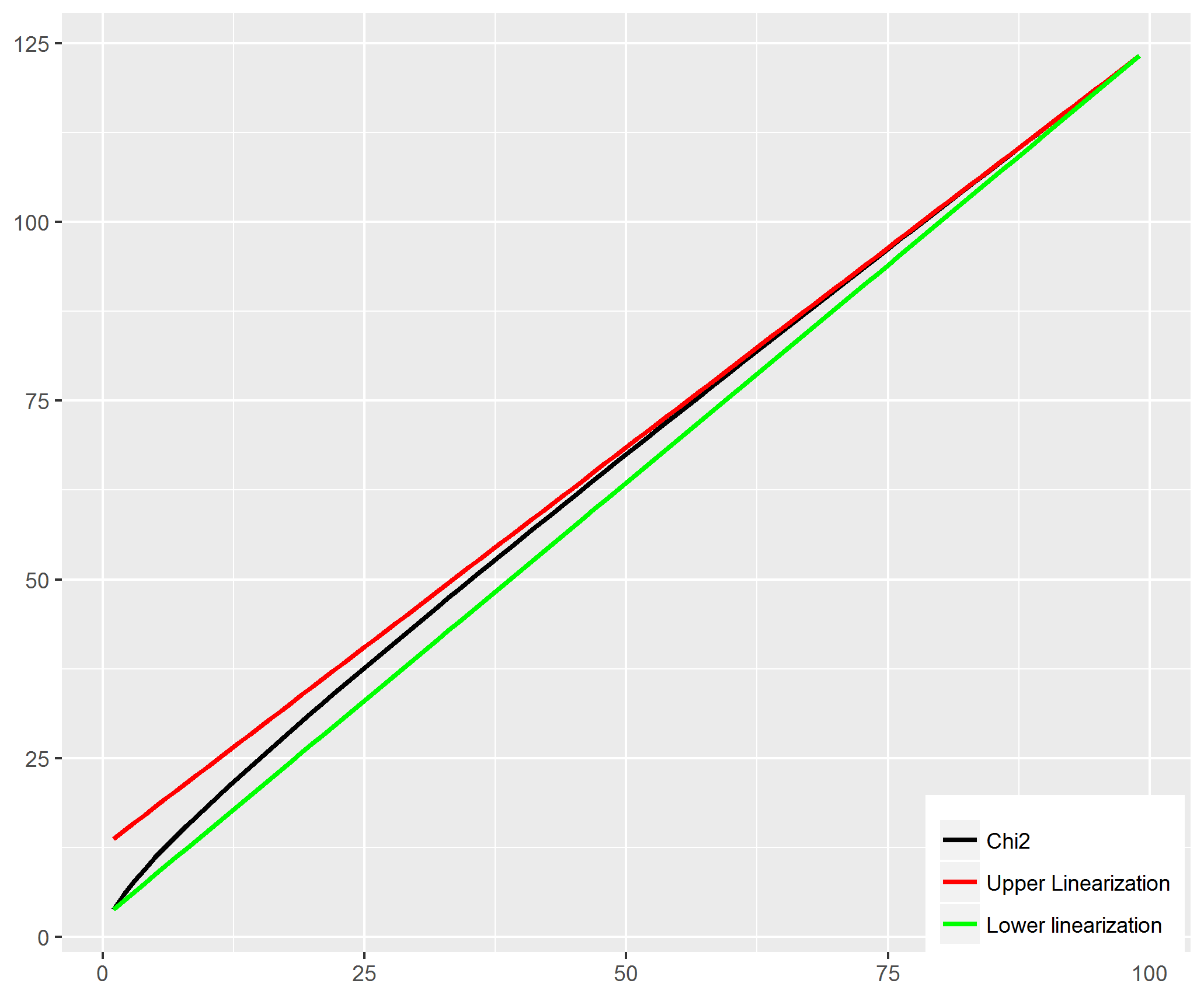}
\caption{Lower and upper approximations for the $\chi^2$ quantile up to $n=100$.}
\label{fig:Chi2Approx}
\end{figure}

More formally, assume we dispose of a sample with $n$ observations. Denote $U_n(x) = a_{U,n} x + b_{U,n}$ the first linearization of the $\chi^2$ quantile which is in fact an upper affine bound of it. Denote also $L_n(x) = a_{L,n} x + b_{L,n}$ another linearization which is a lower affine bound of the $\chi^2$ quantile. The constants are given by
\begin{eqnarray}
a_{U,n} = \mathbb{F}_{\chi^2(n-1)}^{-1}(1-\alpha) - \mathbb{F}_{\chi^2(n-2)}^{-1}(1-\alpha), & \quad & b_{U,n} = \mathbb{F}_{\chi^2(n-1)}^{-1}(1-\alpha) - a_{U,n}(n-1) \label{eqn:UpperApprox}\\
a_{L,n} = \frac{\mathbb{F}_{\chi^2(n-1)}^{-1}(1-\alpha) - \mathbb{F}_{\chi^2(1)}^{-1}(1-\alpha)}{n-2},& \quad & b_{L,n} = \mathbb{F}_{\chi^2(1)}^{-1}(1-\alpha) - a_{L,n},
\label{eqn:LowerApprox}
\end{eqnarray}
and an example with $n=100$ is shown in figure (\ref{fig:Chi2Approx}). These constants ensure that the upper affine bound is the tangent on the curve of the $\chi^2$ quantile on $n-1$, so that it has greater values than all the quantile values required for testing the partitions. For the lower affine bound, we make sure that the line starts at the quantile with $n-1$ degrees of freedom and ends at the quantile with 1 degree of freedom so that it remains below all required quantiles for testing the partitions. \\
\emph{Why should it work?} We first need to precise what do we mean by the contribution of a block in a hypothesis. Let's imagine a very simple situation where the tested partition is decomposed into two blocks, that is 
\[H = \{\mu_1=\cdots=\mu_k<\mu_{k+1}=\cdots=\mu_n\}.\]
Using the LR test, we compare the LR statistic with the quantile of the $\chi^2(n-2)$, say $\gamma_{n-2}$. If
\begin{equation}
LR(\{\mu_1=\cdots=\mu_k<\mu_{k+1}=\cdots=\mu_n\}) \leq \gamma_{n-2}
\label{eqn:LRAllApprox}
\end{equation}
the hypothesis $H$ is not rejected. The LR statistic is the sum of two terms, namely
\[LR(\{\mu_1=\cdots=\mu_k<\mu_{k+1}=\cdots=\mu_n\}) = LR(\{\mu_1=\cdots=\mu_k\}) + LR(\{\mu_{k+1}=\cdots=\mu_n\}).\]
Now, replacing $\gamma_{n-2}$ by, for example, the upper affine bound in inequality (\ref{eqn:LRAllApprox}) gives
\[LR(\{\mu_1=\cdots=\mu_k\}) - a_{U,n} (k-1) + LR(\{\mu_{k+1}=\cdots=\mu_n\}) - a_{U,n}(n-k-1) \leq b_{U,n}.\]
Define the contribution of the block $\mu_1=\cdots=\mu_k$ by $LR(\{\mu_1=\cdots=\mu_k\}) - a_{U,n} (k-1)$ (and similarly for the other one). Now, each block contributes in judging on the (non)rejection of the partition $H$ by a quantity. The sum of contributions produces the contribution of the whole partition in the judgment on its (non)rejection, because by summing the two contributions it only remains to compare the sum with the intercept $b_{U,n}$. If it is greater, then we reject, otherwise we do not reject.\\
Suppose now that we are interested in the first center, so that the block of interest is $\mu_1=\cdots=\mu_k$. If we have an idea about the minimum contribution that an additional block from the subset of centers $\{\mu_{k+1},\cdots,\mu_{n}\}$ (for example $\mu_{k+1}=\cdots=\mu_{n}$) can add, then we will calculate the sum
\[LR(\{\mu_1=\cdots=\mu_k\}) - a_{U,n} (k-1) + \min_{H_i\in \mathcal{P}(\{\mu_{k+1},\cdots,\mu_{n}\})} \{LR(H_i) - a_{U,n}(\#H_i)\}\]
where $\mathcal{P}(\{\mu_{k+1},\cdots,\mu_{n}\})$ is the set of all (correctly ordered) sub-partitions that can be generated from the set of centers $\{\mu_{k+1},\cdots,\mu_{n}\}$ by introducing "$=$" or "$<$" between the centers, and $\#H_i$ is the number of equalities forming this sub-partition. Then, it suffices to compare this sum of two contributions with $b_{U,n}$. If it exceeds it, then the block of interest $\mu_1=\cdots=\mu_k$ is rejected, otherwise it is not rejected and is taken into account in the calculus of the confidence intervals.\\
Last but not least, if we use the upper affine bound instead of the $\chi^2$ quantiles, then we will \emph{reject at most} the same number of partitions that we reject by using the $\chi^2$ quantiles. Thus, the obtained confidence intervals are more conservative than the ones obtained using the $\chi^2$ quantiles and thus construct upper bounds for the confidence intervals obtained from the exact partitioning (obtained using either Algorithm \ref{algo:LevelLevel} or Algorithm \ref{algo:BlockPartitioning}). On the other hand, with similar arguments, using the lower affine bound we \emph{reject at least} as much as we reject using the $\chi^2$ qunantiles and we, therefore, construct lower bounds for the confidence intervals obtained from the exact partitioning (obtained using either Algorithm \ref{algo:LevelLevel} or Algorithm \ref{algo:BlockPartitioning}).\\

\emph{How do we calculate the minimum contribution of a subset of centers in a polynomial time?} It is done iteratively starting from the smallest blocks. Assume we want to calculate the minimum contribution of a subset of centers, say $\{\mu_1,\cdots,\mu_k\}$. We need at first to use previously calculated contributions, that is the contribution of the subsets $\{\mu_1,\cdots,\mu_i\}$ and for the subsets $\{\mu_i,\cdots,\mu_k\}$ for $i = 2,\cdots,k-1$. In order to keep things clear, we adapt the following notation
\[\text{minContrib}(\{\mu_{1},\cdots,\mu_{k}\}) = \min_{H_i\in \mathcal{P}(\{\mu_{1},\cdots,\mu_{k}\})} \{LR(H_i) - a_{U,n}(\#H_i)\}.\]
We have:
\begin{multline*}
 \text{minContrib}(\{\mu_{1},\cdots,\mu_{k}\}) = \min\left[LR(\{\mu_1=\cdots=\mu_k\}) - a_{U,n} (k-1),\; \right.\\ \text{minContrib}(\{\mu_{1},\cdots,\mu_{k-1}\}),\;   \text{minContrib}(\{\mu_{2},\cdots,\mu_{k}\}),\; \\ \left. \left\{\text{minContrib}(\{\mu_{1},\cdots,\mu_{i}\}) + \text{minContrib}(\{\mu_{i+1},\cdots,\mu_{k}\}), i=2,\cdots,k-1\right\} \right].
\end{multline*}
In other words, the minimum contribution of a block is the minimum between several terms:
\begin{itemize}
\item the contribution of the whole block of centers;
\item the minimum contribution of the subset of centers after having excluded either of the borders (previously calculated);
\item the sum of contributions of any two subsets of the form $\{\mu_1,\cdots,\mu_{i}\},\{\mu_{i+1},\cdots,\mu_k\}$ for all values $i=2,\cdots,k-1$. Notice that since the two subsets of centers are disjoint, this corresponds to the minimum contribution of the union $\{\mu_1,\cdots,\mu_{i}\}\cup\{\mu_{i+1},\cdots,\mu_k\}$. Moreover, the minima contributions of the subset $\{\mu_1,\cdots,\mu_{i}\}$ and the subset $\{\mu_{i+1},\cdots,\mu_k\}$ should have already been calculated in a previous step. 
\end{itemize}
Algorithm \ref{algo:MinContrib} shows a pseudo-code of these steps.\\
\begin{algorithm}
\KwData{Ordered sample $y_1,\cdots,y_n$ and standard deviations $\sigma_1,\cdots,\sigma_n$.
}
\KwResult{For each block $\{\mu_i,\cdots,\mu_j\}$, its minimum contribution.}
minContrib = matrix(0, nrow = n, ncol = n)\;
\For{j from 2 to n}
{
 \For{i from j-1 to 1}
 {
  Average = $\frac{1}{\sum_{l = i}^j{1/\sigma_l^2}}\sum_{l = i}^j{y_l/\sigma_l^2}$ \;
	$LogL = \sum_{s = i}^j{(y_s - \text{Average})^2/\sigma_s^2}$\;
	minContrib[i,j] = $\min\left(\text{minContrib}[i,i:(j-1)]+\text{minContrib}[(i+1):j,j], LogL - a_{n,U}(j-i)\right)$
 }
}
\label{algo:MinContrib}
\caption{Calculating the matrix of minimum contributions}
\end{algorithm}
Algorithm \ref{algo:MinContrib} must be repeated to calculate also the minimum contributions for the case of the lower affine bound. Finally, the lower and upper bounds for the confidence intervals are calculated by going through the centers and by testing all possible blocks starting at each one of these centers. Notice that the complexity of our new approach is $\mathcal{O}\left(n^3\right)$.
\begin{remark}[Exact partitioning]\label{rem:ExactPartition}
If both the lower-bound and upper-bound confidence intervals coincide on some of the centers, then exact-partitioning confidence intervals for these centers are the shared confidence intervals. If we are lucky enough and both the lower- and upper-bound confidence intervals coincide for all the centers, then we would have found the result of the exact partitioning. This can occur if we have a large number of centers which are close to each other so that the confidence intervals are wide enough and only hypotheses with large blocks are interested in the partitioning. For these hypotheses, we are situated at the top right part of figure (\ref{fig:Chi2Approx}), and the approximation of the $\chi^2$ quantile by the two lines is the most accurate.
\end{remark}
\begin{remark}
\label{rem:BestMinMaxBlockSize}
It is possible to extract from the lower and upper bounds of the confidence intervals lower and upper bounds for the size of a block that need to be tested in Algorithm \ref{algo:BlockPartitioning} (the \texttt{MinBlock} and \texttt{MaxBlock} vectors). Indeed, consider the center $\mu_i$ for which we got the lower-bound confidence interval $[a_i,b_i]$, then we know for sure that there exists at least one partition which contains the block $\mu_i=\cdots=\mu_{b_i}$ which was not rejected, and thus we should at least check larger blocks giving a minimum block size of $b_i-i$ for center $\mu_i$. On the other hand, assume that we get an upper-bound confidence interval equal to $[c_i,d_i]$, then we know for sure that the block $\mu_i=\cdots=\mu_{d_i+1}$ was rejected in all partitions containing this block. Thus, we should only check partitions containing smaller blocks giving a maximum block size of $d_i - i$. Of course, we do not need to consider centers for which the lower- and upper-bound confidence interval coincide according to the previous remark. 
\end{remark}
\begin{remark}[anti-conservative]
Although the use of a lower affine bound of the critical value permits to obtain lower bounds for the confidence intervals, these confidence intervals might not be anti-conservative confidence intervals after all. Indeed, According to what we presented in Section \ref{sec:CriticalLevel}, even the $\chi^2$ critical values may produce conservative confidence intervals for the ranks and more adjustments on the critical value might be possible. Therefore, the lower-bound confidence intervals might also be conservative w.r.t. the \emph{exact} confidence intervals of the ranks using the exact quantile (which is difficult or impossible to calculate). The calculus of an exact quantile was discussed in paragraph \ref{subsec:AdjustCritVal} and remains an open problem for possible future discussions.
\end{remark}
\subsection{Further improvements on the lower and upper bounds}
There is no doubt that our lower and upper approximations for the $\chi^2$ quantiles are not the best options for a specific situation. However, they seem to be very good options for a general setup. For a specific purpose and a given sample, it might be possible to improve upon the lower and upper bounds so that we get an even narrower bracketing of the partitioning's exact result. Suppose that we have a large $n$, say a hundred centers. We have a better chance to obtain narrow bracketing (or even exact partitioning, see remark \ref{rem:ExactPartition}) if the size of the confidence intervals are very wide, simply because testing large blocks requires a greater number of degrees of freedom. This implies that we will be using the lower and upper linearizations near $n=100$ where the two lines are the closest to each others and to the curve of the $\chi^2$ quantiles. It is very probable that both the lower and the upper approximations yield the same confidence intervals so that we get the exact partitioning. On the other hand, if the size of the confidence intervals is small, then we will be using the linearizations near the 1 degree of freedom where the two lines are the most far from each others, and small adjustments on the lower linearization may have an impact on the final result. Therefore, we suggest to adapt the linearizations according to the sample.
\paragraph{Adjustment on the lower bound.} An automatic way to do it is to start with the results of the bracketing we presented at the beginning of this section. We calculate the minimum block size that was not rejected but this time for all the centers together and not center by center, that is the minimum of the \texttt{MinBlock} vector in Algorithm \ref{algo:BlockPartitioning}. In case we use Remark \ref{rem:BestMinMaxBlockSize} to calculate the \texttt{MinBlock} vector, then we need to take into account that the last centers must not be taken into consideration in the calculus of this minimum, and we must stop at the first center whose confidence interval contains the last center $\mu_n$. Indeed, Using Remark \ref{rem:BestMinMaxBlockSize}, for the last centers the blocks sizes start to decrease automatically since they are bounded by the last center $\mu_n$ and the block size for the last center is 1. Denote the obtained minimum size $n_0$. Now, we adjust the slop and the intercept of the lower approximation as follows
\begin{equation}
a_{L,n} = \frac{\mathbb{F}_{\chi^2(n-1)}^{-1}(1-\alpha) - \mathbb{F}_{\chi^2(n_0-1)}^{-1}(1-\alpha)}{n-n_0-2}, \quad  b_{L,n} = \mathbb{F}_{\chi^2(1)}^{-1}(1-\alpha) - a_{L,n}(n_0-1)
\label{eqn:LowerApproxAdjusted}
\end{equation}
\begin{figure}[ht]
\centering
\includegraphics[scale = 0.6]{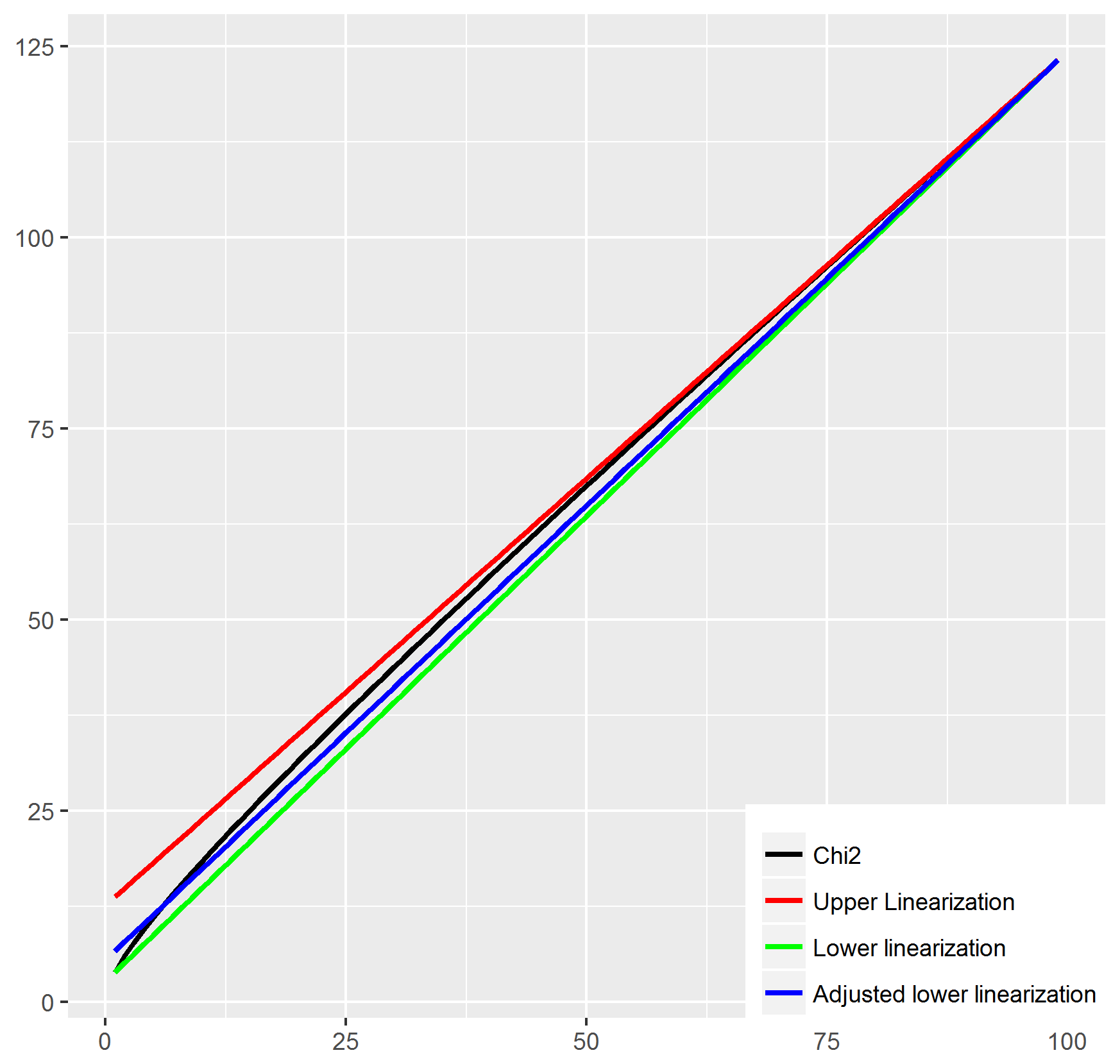}
\caption{Lower and upper approximations for the $\chi^2$ quantile up to $n=40$. Adjustment calculated based on a dataset of 40 observations for which the minimum block size was equal to 5.}
\label{fig:Chi2ApproxAdjust}
\end{figure}
In figure (\ref{fig:Chi2ApproxAdjust}), we show an example of this adjustment. We tried it out on a 40-sample for which the minimum block size found using the lower bound (\ref{eqn:LowerApprox}) was 5. The number of centers where the lower- and upper-bound confidence intervals coincide (and equal to the exact partitioning) is equal to 35 before the adjustment. It became 37 after the adjustment. It means there are only 3 centers left for which the confidence intervals need to be adjusted using Algorithm \ref{algo:BlockPartitioning} in order to get exact-partitioning confidence intervals.
\begin{remark}
Since the lower and upper bounds coincide in general at a good portion of the centers, their confidence intervals must not be touched when we adjust the lower bound. The minimum block size needed in order to find the starting point of the lower affine approximation should only take into account the centers for which the lower and upper bounds are different.
\end{remark}
\paragraph{Adjustments on the upper bound.} We constructed the upper affine bound so that it is the tangent on the curve of the $\chi^2$ at point $n-1$. Since for large values of $n$, more than half the curve of the $\chi^2$ quantile is almost linear, we could move the tangent point from $n-1$ to say $\lfloor n/2\rfloor$. The changes in the results for the larger numbers of degrees of freedom are minor in comparison to the smallest ones. Moreover, any choice for the tangent point provides us with a new affine upper bound and new upper bounds for the confidence intervals. It is thus possible to try out several clever ones and see at first which one provides us with more intersections with the lower-bound confidence intervals (see Remark \ref{rem:ExactPartition}). Then, we can look for other points of intersection occurring with other choices for the tangent points. These can then be integrated in the initial upper-bound confidence intervals so that we make as many corrections as possible.
\begin{figure}[ht]
\centering
\includegraphics[scale = 0.6]{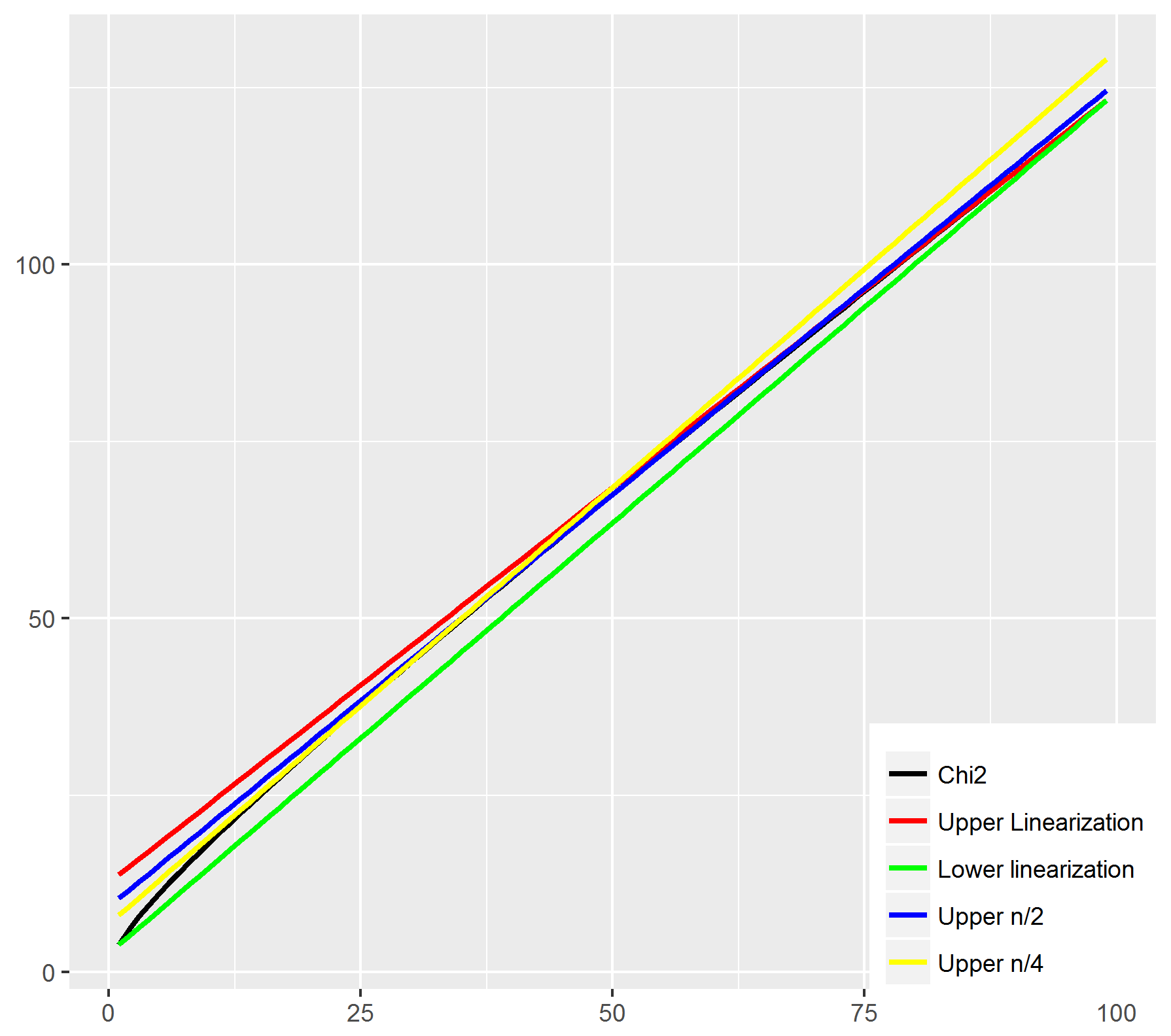}
\caption{Several affine upper approximations by changing the tangent point from $n$ to $n/2$ and finally to $n/4$.}
\label{fig:Chi2UppApproxAdjust}
\end{figure}
\begin{example}[simulation]
Our affine bracketing of the Chi square quantiles introduced in Section \ref{sec:ApproxPartition} seems to be a powerful tool. Most of the time, the lower and upper bounds coincide at least at $50\%$ of the centers so that the exact-partitioning confidence intervals are reached directly. In the simulated examples we show here in this paper, the approximation produced most of the time more than $80\%$ of the exact-partitioning confidence intervals. Surprisingly, the remaining confidence intervals were found with a maximum error of one rank per center. The most difficult situations, where the gap between the lower and upper bound may increase, occur when the proportion of unrejected partitions whose number of degrees of freedom related to the LR test is small. This is because the gap between the lower- and upper-affine bounds is the largest for the smallest degrees of freedom. This happens for example if we have a large number of centers which are distant from each others. The larger $n$, the larger the gap near small degrees of freedom becomes and the bracketing becomes wider. On the other hand, the larger the distance between the centers, the more often we easily reject partitions with large blocks of equalities and do not reject very small blocks. The small blocks are tested with the LR at small degrees of freedom where the gap between the bounds is the largest. In figure (\ref{fig:WorstCaseBounds}), we illustrate a situation where the true centers are the numbers $\{0,3,6,9,\cdots,300\}$. We generated a 100-sample from Gaussian distributions with standard deviations equal to 1. We can notice that the confidence intervals are very small and the lower- and upper-affine bounds hardly coincide. The number of confidence intervals where the lower and upper-affine bounds produce the same confidence intervals is 30 centers (30 percent). Besides, for some of the centers the difference between the lower- and the upper-bound confidence intervals is 2. \\
\begin{figure}[ht]
\centering
\includegraphics[scale = 0.7]{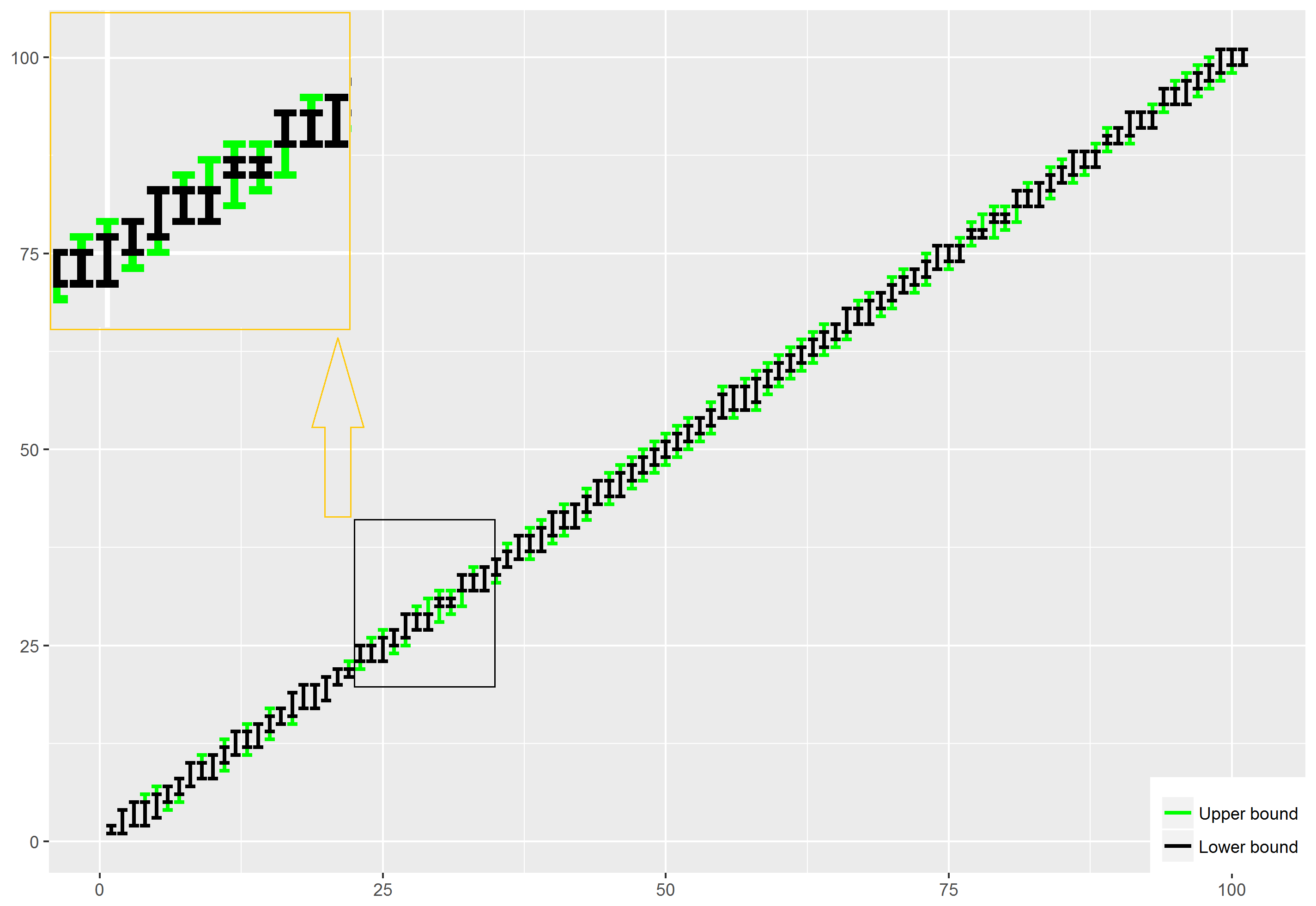}
\caption{An example with bad approximation by the lower- and upper-affine bounds of the Chi square quantile. The maximum error is 2 which happens at several centers. Green CIs are the upper-bound CIs and the black ones are the lower-bound CIs.}
\label{fig:WorstCaseBounds}
\end{figure}
\emph{Corrections on the upper-affine bound.} In the situation of figure (\ref{fig:WorstCaseBounds}), no correction on the lower-affine bound can be achieved because the minimum block-size is equal to 1. We can however adjust the upper-affine bound by changing the tangent point to $n/2 = 50$. This increases the proportion of confidence intervals where the lower- and upper-affine bounds produce the same confidence intervals to $45\%$. By making the tangent point at $n/4 = 25$, we get a proportion of $52\%$. There are two interesting facts about this example. Firstly, this is somehow a worst-case scenario which we may encounter in practice and we got a gap of 2 between the lower- and upper-bound confidence intervals. The second one is that by using an upper-affine bound tangent on the $\chi^2$ quantile curve at $n/4$, we get to improve the overall matching (where an exact result for the partitioning is obtained), and on the other hand the gap of 2 between the lower- and upper-bound confidence intervals disappears.\\
\emph{Correction on the lower-affine bound and exact partitioning.} In figure (\ref{fig:EquidistanceBounds}), we show a very basic example with 20 centers. We also added the confidence intervals produced by the exact partitioning. The exact-partitioning confidence intervals coincide with the upper-bound confidence intervals. The adjustment on the lower-affine bound (\ref{fig:Chi2ApproxAdjust}) improved the results for 4 centers out of 20 centers. The adjustment using formulas (\ref{eqn:LowerApproxAdjusted}) is done with $n_0 = 4$. Surprisingly, after the adjustment, the lower- and upper-bound confidence intervals coincide and the exact-partitioning confidence intervals are found for all the centers.
\begin{figure}[ht]  
\centering
\includegraphics[scale = 0.7]{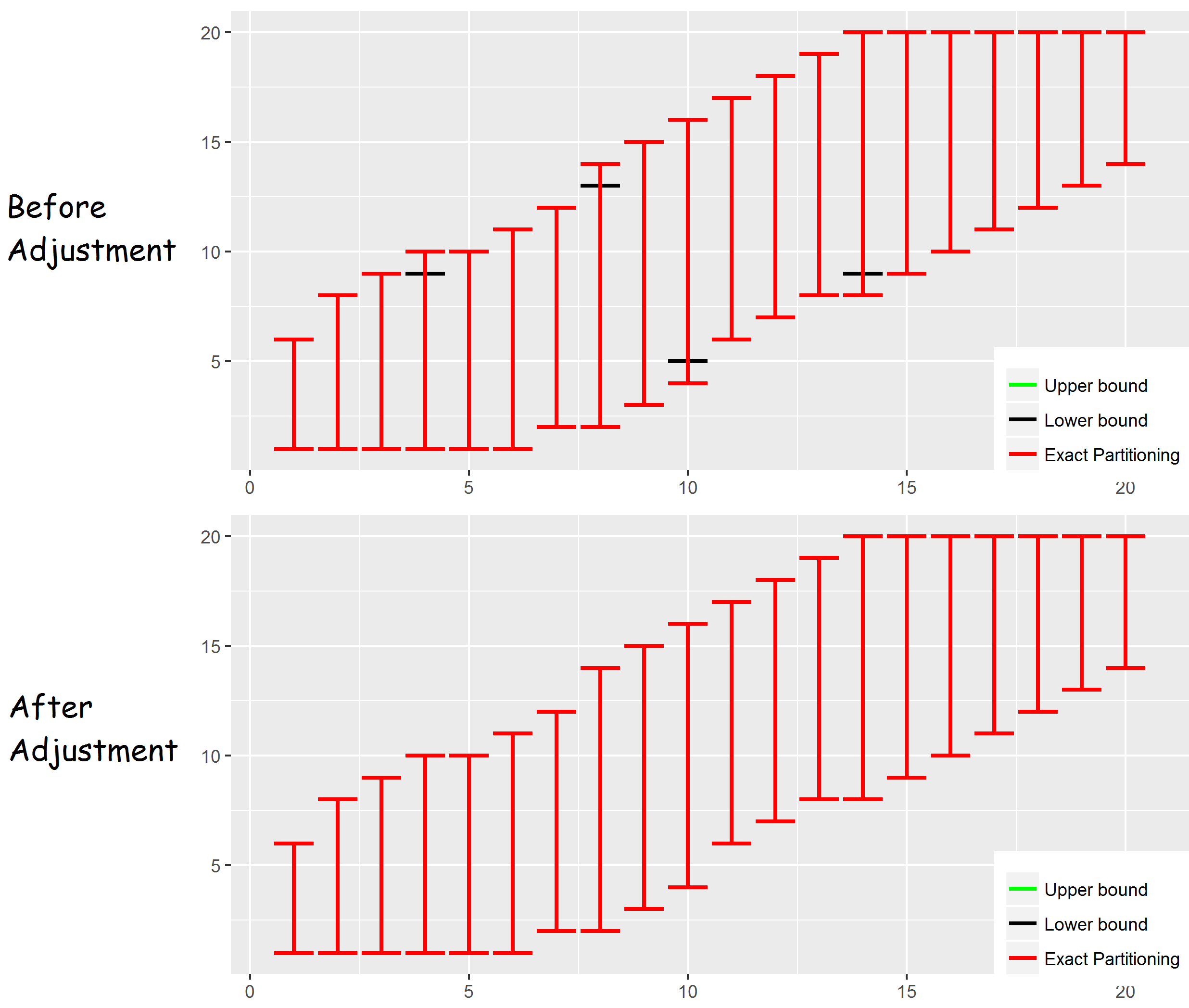}
\caption{Exact partitioning and the lower and upper bounds. An improvement on the lower bounds were made for two centers by adjusting the lower-affine bound.}
\label{fig:EquidistanceBounds}
\end{figure}
\end{example}

\section{Tukey's pairwise comparison}\label{sec:TukeyHSD}
Tukey's pairwise comparison procedure (\cite{Tukey}) well-known as the Honest Significant Difference test (HSD) is an easy way to compare observations or estimates having (assumed) Gaussian distributions especially in ANOVA models. The interesting point about the procedure is that it provides simultaneous confidence statement and controls the FWER at level $\alpha$. Moreover, it possesses optimality properties in some contexts. For example, in a balanced one-way layout, among all procedures that give equal width intervals  for all pairwise differences with joint level superior to $1-\alpha$, the HSD gives the shortest intervals, see \cite{HochbergBook} page 81 for more details. Easiness and optimality are the two reasons behind our interest in this method so that our method presented in the first part of this paper can be compared to it.
\subsection{The method} 
Suppose that $y_1,\cdots,y_n$ is a Gaussian sample generated from the Gaussian distributions $\mathcal{N}(\mu_i,\sigma_i^2)$. In order to produce a simultaneous confidence intervals for the ranks of the means $\mu_1,\cdots,\mu_n$, we test all hypotheses of the form $\mu_i=\mu_j$ using the following rejection region
\begin{equation}
\left\{\frac{|y_i-y_j|}{\sqrt{\sigma_i^2 + \sigma_j^2}}>q_{1-\alpha}\right\}
\label{eqn:RejRegionTukey}
\end{equation}
where $q_{1-\alpha}$ is the quantile of order $1-\alpha$ of the distribution of the Studentized range
\begin{equation}
\max_{i,j=1,\cdots,n}\frac{|Y_i-Y_j|}{\sqrt{\sigma_i^2 + \sigma_j^2}},
\label{eqn:StudentizedRange}
\end{equation}
and $Y_i$ and $Y_j$ are two centered Gaussian random variables with standard deviations $\sigma_i$ and $\sigma_j$ respectively. The confidence interval for the rank of center $\mu_i$, say $[r_{i,L},r_{i,U}]$ is calculated by counting how many hypotheses $H_{i,j}$ was rejected and such that $y_j<y_i$ (which yields $r_{i,L}$). Then we calculate how many hypotheses $H_{i,j}$ was not rejected and such that $y_j>y_i$ (which yields $n-r_{i,U}$). Using Tukey's HSD to construct confidence intervals for ranks was recently proposed by \cite{OurTukeyPaper}. Although other approaches (step-down algorithms) were considered in that paper, we are only interested in Tukey's HSD in order to keep the ideas clear.\\

\subsection{A new look at Tukey's pairwise comparison as a partitioning technique} 
It is possible to define a statistical (local) test over the partitions related to the confidence intervals problem presented earlier in paragraph \ref{sec:HowToRank} which yields the same confidence intervals for the ranks as Tukey's HSD. In order to present this test, we need to introduce new notations. We will write a hypothesis $H_i$ as a set of blocks where each block contains centers related by an equality. For example, we write $H_i = \{A=B=C<D<E=F<G\}$ as $H_{i,1}\cup H_{i,2}\cup H_{i,3}\cup H_{i,4}$ with $H_{i,1}=\{A=B=C\}, H_{i,2}=\{D\}, H_{i,3}=\{E=F\}, H_{i,4}=\{G\}$. \\
The first step in writing Tukey's HSD as a partitioning procedure is to calculate for each block of centers the maximum and the minimum observed values. If the observed maximum of a block (calculated using the $y_i$'s) is larger than the observed minimum of the block coming after it in the hypothesis, then the two blocks are combined (pooled) into one. Denote the new block $\tilde{H}_{i,j}$. The cardinal of $H_i$ denoted $\# H_i$ will be considered as the number of remaining blocks after all. The second step is to test the hypotheses $H_i$'s using the following rejection region
\begin{equation}
\left\{\max_{k=1,...,\# H_i}\max_{y_j\in \tilde{H}_{i,k}} \frac{|y_k-y_j|}{\sqrt{\sigma_k^2 + \sigma_j^2}} > q_{1-\alpha}\right\},
\label{eqn:RejRegionPartitionTukey}
\end{equation}
where $y_k$ corresponds to the smallest observed value related to the centers inside $\tilde{H}_{i,k}$, and $q_{1-\alpha}$ is the quantile of order $1-\alpha$ of the Studentized range (\ref{eqn:StudentizedRange}) as in Tukey's HSD procedure. Note that we use the same critical value for all the partitions from any level of the partitioning scheme.\\
In order to prove that this partitioning procedure produces the same confidence intervals as Tukey's HSD, we first need the following Lemma.
\begin{lemma}\label{lem:TukeyShortcut}
As in the case of the LR test, it suffices in the above partitioning procedure to test only correctly ordered hypotheses, that is the hypotheses whose ordering coincides with the empirical one.
\end{lemma}
\begin{proof}
Let $H_i$ be one of the partitioning hypotheses whose ordering of the centers does not comply with the empirical one. Without loss of generality, suppose that $H_i=\{H_{i,1},H_{i,2},H_{i,3}\}$. Suppose that the empirical ordering says that $\max_{y_i, s.t. \mu_i\in H_{i,1}} H_{i,1}>\min_{y_i, s.t. \mu_i\in H_{i,2}} H_{i,2}$, then our testing procedure will pool these two blocks into one $\tilde{H}_{i,1}$. In the same spirit of the proof of Proposition \ref{lemm:ShortcutCorrectHyp}, if $H_i$ is rejected, this changes nothing in terms of the confidence intervals and we only need to look at the unrejections. 
Suppose now, that $H_i$ is not rejected, then
\begin{equation}
\max_{y_j\in\tilde{H}_{i,1}}\frac{|y_j-y_{i_1}|}{\sqrt{\sigma_{i_1}^2 + \sigma_j^2}}\leq q_{1-\alpha}, \text{ and } \max_{y_j\in H_{i,3}}\frac{|y_j-y_{i_3}|}{\sqrt{\sigma_{i_3}^2 + \sigma_j^2}} \leq q_{1-\alpha}
\label{eqn:TukeyNotOrderAccept}
\end{equation}
where $y_{i_1}$ and $y_{i_3}$ correspond to the smallest observed values related to the centers inside $\tilde{H}_{i,1}$ and $H_{i,3}$ respectively. The hypothesis $\{H_{i,1}\cup H_{i,2}, H_{i,3}\}$ is one of the partitions from an upper level whose ordering coincides with the empirical one. Besides, this hypothesis is not rejected due to (\ref{eqn:TukeyNotOrderAccept}) because on the one hand, it has the same test statistic as $H_i$ and on the other hand, it has the same common critical value. Thus, for any hypothesis $H_i$ with incorrect ordering, there exists a correctly ordered hypothesis which has the same test statistic and whose unrejection implied by the unrejection of the incorrectly ordered one. Since these are the only partitions needed in the calculus of the confidence intervals, then this completes the proof.
\end{proof}
We may now state our result.
\begin{proposition}
The partitioning procedure defined using the rejection region (\ref{eqn:RejRegionPartitionTukey}) is equivalent to Tukey's pairwise comparison procedure defined through the rejection region (\ref{eqn:RejRegionTukey}). In other words, they produce the same simultaneous confidence intervals for the ranks of the centers $\mu_1,\cdots,\mu_n$.
\end{proposition}
\begin{proof}
Due to Lemma \ref{lem:TukeyShortcut}, we only need to consider correctly ordered hypotheses. The rejection region for these hypotheses turns out to be a calculus of the maximum of the maxima differences inside the blocks composing the hypothesis. Take center $\mu_i$. Suppose that with Tukey's procedure, we determined a confidence interval for the rank of $\mu_i$ to be $[L_i, U_i]$. This means that we could not reject all hypotheses $\mu_i=\mu_j$ for $j\in[L_i,U_i]$. In other words, we have:
\[\frac{|y_i - y_j|}{\sqrt{\sigma_i^2 + \sigma_j^2}}\leq q_{1-\alpha},\quad \forall j \in [L_i,U_i].\]
Besides, we rejected all hypotheses $\mu_i = \mu_l$ for $l\leq L_i-1$ and $l\geq U_i + 1$. In other words
\[\frac{|y_l - y_i|}{\sqrt{\sigma_i^2 + \sigma_l^2}} > q_{1-\alpha}, \quad \forall l \in \{1,\cdots,L_{i}-1\}\cup\{U_i + 1,\cdots,n\}.\]
Let us check what is the confidence interval that we can get using the partitioning with (\ref{eqn:RejRegionPartitionTukey}) from these rejections and non rejections. First of all, we have 
\begin{eqnarray*}
\frac{y_{U_i+1} - y_i}{\sqrt{\sigma_{U_i+1}^2+\sigma_i^2}} & > & q_{1-\alpha} \\
\frac{y_i - y_{L_i-1}}{\sqrt{\sigma_i^2 +\sigma_{L_i-1}^2}} & > & q_{1-\alpha}
\end{eqnarray*}
Thus any partition containing the block $\mu_{i}=\cdots=\mu_{U_i+1}$ or the block $\mu_{L_i-1}=\cdots=\mu_{i}$ (or larger ones) is rejected using the rejection region (\ref{eqn:RejRegionPartitionTukey}). This also entails that any hypothesis producing a larger confidence interval (more equalities) will also be rejected. Therefore, we can conclude for the time being that the confidence interval for $\mu_i$ produced by the partitioning procedure is at most the one produced by Tuky's HSD, that is $[L_i,U_i]$. \\
Suppose now that with the partitioning procedure, we got a confidence interval for $\mu_i$ equal to $[L_P,U_P]$. We are then sure that any hypothesis containing the block $\mu_i=\cdots=\mu_{U_P+1}$ or the block $\mu_{L_P-1}=\cdots=\mu_i$ is also rejected. In particular, the hypotheses $\{\mu_1<\cdots<\mu_i=\cdots=\mu_{U_P+1}<\cdots<\mu_n\}$ and $\{\mu_1<\cdots<\mu_{L_P-1}=\cdots=\mu_i<\cdots<\mu_n\}$ are rejected. This means that
\[\max_{j=i,\cdots,U_P+1} \frac{|y_i- y_j|}{\sqrt{\sigma_i^2+\sigma_j^2}} = \frac{y_{j_1} - y_i}{\sqrt{\sigma_i^2+\sigma_{j_1}^2}} > q_{1-\alpha}, \qquad \max_{j=L_P-1,\cdots,i} \frac{|y_i- y_j|}{\sqrt{\sigma_i^2+\sigma_j^2}} = \frac{y_i - y_{j_0}}{\sqrt{\sigma_i^2+\sigma_{j_0}^2}} > q_{1-\alpha}.\]
for some $j_0\in\{L_P-1,\cdots,i\}$ and $j_1\in\{i,\cdots,U_P+1\}$ verifying
\begin{eqnarray*}
\forall j\in\{i,\cdots,U_P+1\},&\;& \frac{y_{j_1} - y_i}{\sqrt{\sigma_i^2+\sigma_{j_1}^2}}>\frac{|y_i- y_j|}{\sqrt{\sigma_i^2+\sigma_j^2}} \\
 \forall j\in\{L_P-1,\cdots,i\},&\;& \frac{y_i - y_{j_0}}{\sqrt{\sigma_i^2+\sigma_{j_0}^2}} >  \frac{|y_i- y_j|}{\sqrt{\sigma_i^2+\sigma_j^2}}.
\end{eqnarray*}
This entails that with Tukey's procedure, we must reject hypotheses $\mu_i=\mu_{j_1}$ and $\mu_{j_0}=\mu_i$. Thus, the confidence interval provided by Tukey's procedure is at most the confidence interval produced by the partitioning, that is $[L_P,U_P]$.\\
We proved that Tukey's pairwise procedure cannot produce larger confidence intervals than the partitioning procedure using (\ref{eqn:RejRegionPartitionTukey}), and that the latter cannot produce larger confidence intervals than Tukey's HSD. Hence, Both methods are equivalent, that is they produce the same simultaneous confidence intervals for the ranks.

\end{proof}

\section{Simulation study}\label{sec:simulations}
In this simulation study, we compare the likelihood ratio test and Tukey's HSD. Both methods are corrected for multiple testing according to the partitioning principle, so both methods protect against the FWER at level $\alpha$ and produce simultaneous confidence intervals. Moreover, they both enjoy the power properties of the partitioning principle as a procedure for multiple testing. Recall that the partitioning principle is as powerful as the closed testing procedure and can sometimes perform better (see Section \ref{sec:PartitionPrincip}). Tukey's procedure has its own optimal properties in some contexts (\cite{HochbergBook} page 81) and serves as a good competitor.\\ 
We simulate datasets from Gaussian distributions for two situations. First, we simulate from Gaussian distributions with the same standard deviation (equal to 1) and change the range of the centers so that we understand better the behavior of the two approaches as the difference between the centers increases. After that, we make changes on the standard deviations in order to understand the influence of having different standard deviations for the centers on both compared methods. The goal of course is to identify the situations where either of the methods is a better choice than the other one.\\
The number of considered centers is 100. Algorithms \ref{algo:LevelLevel} and \ref{algo:BlockPartitioning} cannot perform an exact partitioning in a reasonable time. Therefore, the lower and upper bounds (\ref{eqn:LowerApprox},\ref{eqn:UpperApprox}) are used to produce the confidence intervals. As a result, at least $60\%$ of the confidence intervals were found exactly using the lower and upper bounds (they coincided). This increased to more than $90\%$ in the case of centers with a small range. It is interesting to notice that the gap between the lower and the upper confidence intervals per center never exceeded 1 in all these simulations which is a negligible error taking into account that the smallest confidence interval length is of the order of 20 all over the simulations.

\subsection{Equal sigma case}
We fix the standard deviations at 1 and vary the maximum distance between the centers inside one group of centers uniformly distributed. More formally, we have $\mu_i \sim \mathcal{U}(0,\text{Range})$ where Range $\in\{5,10,20,40\}$. Then the observations are generated from $\mathcal{N}(\mu_i,1)$. In figure (\ref{fig:SmallDiffEquidistanceBounds}), we notice that when the range of the data is very small, Tukey's HSD can hardly distinguish any difference between the centers which is not the case of the LR test which is still able to identify some differences between the smallest centers and the largest centers. When the range of the data increases, the difference between Tukey's HSD and the LR test becomes irrelevant until Tukey's HSD gives uniformly slightly shorter confidence intervals than the LR when the range becomes 40. The maximum difference in the ranks when the range is 40 does not exceed 2 per center which also holds for larger ranges. In other words, the gap between the two methods do not keep increasing as the range increases.
\begin{figure}[ht]
\centering
\includegraphics[scale = 0.6]{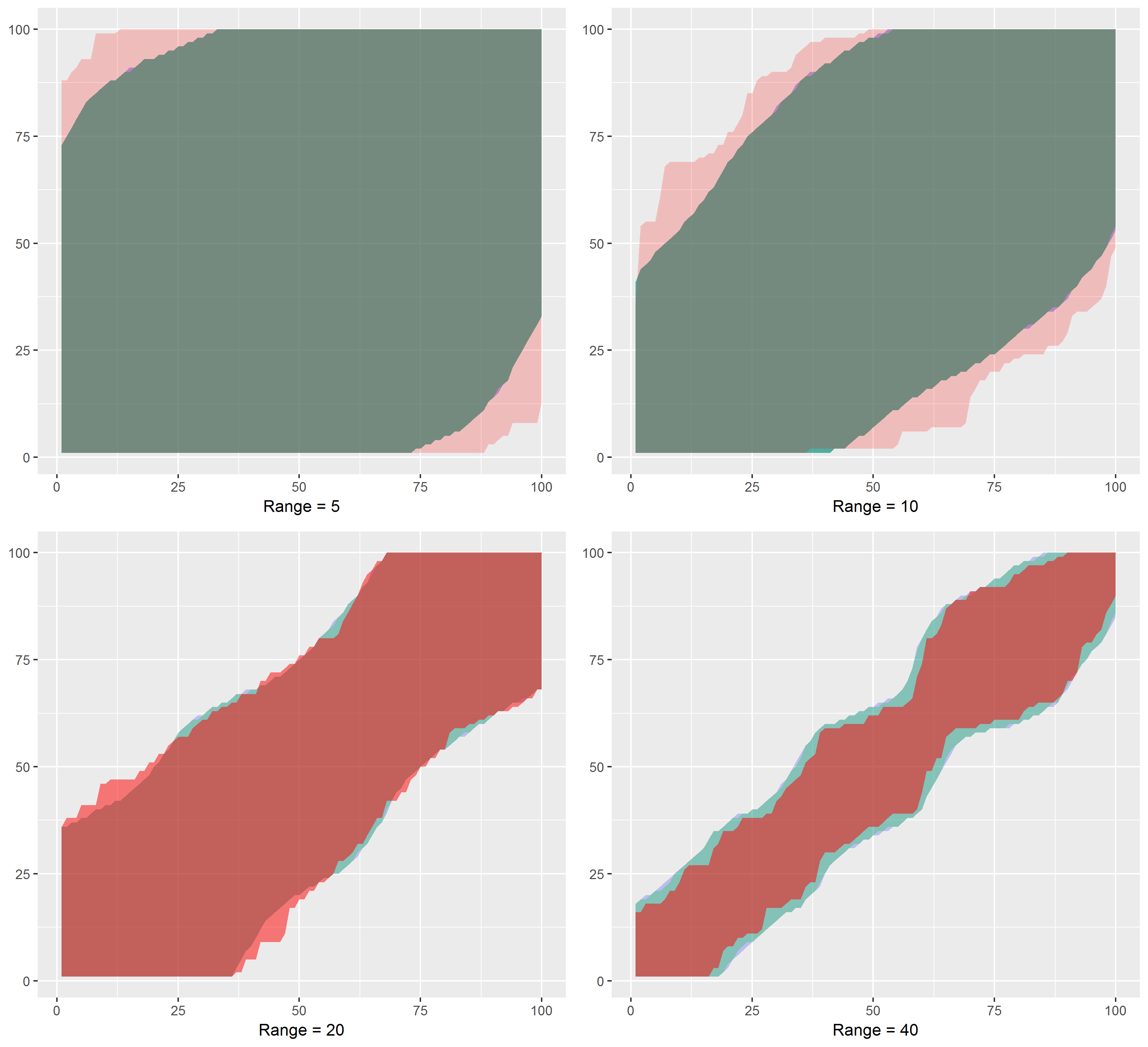}
\caption{The influence of the maximum distance between the centers on the performance of Tukey's pairwise procedure and the LR test. Green CIs are those produced by the upper bound for the partitioning whereas the black CIs are those produced by the lower one.}
\label{fig:SmallDiffEquidistanceBounds}
\end{figure}
\subsection{Different sigma case}
Notice that if the range of the centers (observations) increases, there is no point in trying to increase the values of the standard deviations because both effects will compensate for each others in the LR. It is then important not to change too many factors in the simulation in order to keep the results clear and understandable. We have mainly two factors to play with; the range of the data, and the standard deviations. We will change the range of the data by picking the most extreme situations in the equal-$\sigma$ case (Range = 10, Range = 40). In order to keep things simple, we split the data into two groups where each group of observations comes from a Gaussian. The two Gaussians are $\mathcal{N}(0,0.5)$ and $\mathcal{N}(\text{Range},\sigma^2)$ where Range $\in\{10,40\}$ and $\sigma\in\{1,2,3,4\}$.
It appears in figure (\ref{fig:TwoEqGrpDiffSig10}) that the LR statistics is not greatly influenced by the increase of the standard deviations in the same way as Tukey's HSD. In Tukey's procedure, since the test statistics concerns a pair of observations, the variability of the standard deviation will have a direct impact on the value of the statistics, whereas this variability will just be summed up through several centers in the same time reducing its effect on the test after all.
\begin{figure}[ht]
\centering
\includegraphics[scale = 0.6]{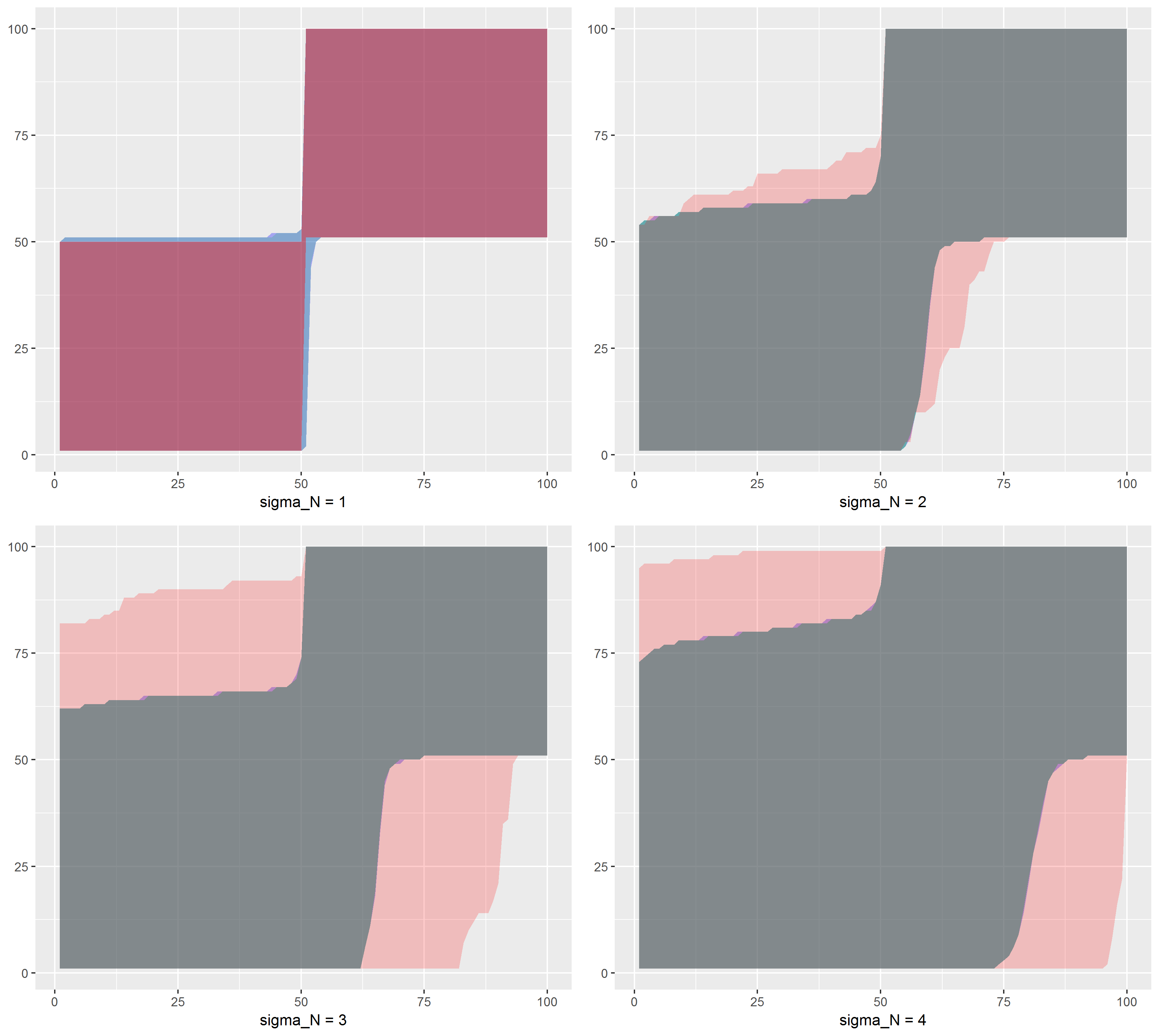}
\caption{The influence of the variability of the standard deviations of the centers on the performance of Tukey's pairwise procedure and the LR test. The maximum distance between the centers is equal to 10. Red CIs are those produced with Tukey's procedure. Green CIs are those produced by the upper bound for the partitioning whereas the black CIs are those produced by the lower one.}
\label{fig:TwoEqGrpDiffSig10}
\end{figure}

\clearpage

\appendix
\section{More details on the probability distribution of the LR statistic under full order constraint}
The objective of this appendix is to give more detailed explanations about how to obtain the probability distribution of the LR statistic for ordered hypotheses with or without equalities under the least favorable situations that is when all the centers are equal. Recall first that this calculus is only valid under the Gaussian model. For other models, some approximations and asymptotic calculus can still be made, but no exact formula can generally be obtained. \\
In order to make things simple, we only consider the case of 3 centers. The LR is given by:
\begin{equation}
LR = \min_{\mu_A,\mu_B,\mu_C\in\mathcal{H}_0} \left\{\frac{1}{\sigma_A^2}(y_A - \mu_A)^2 + \frac{1}{\sigma_B^2}(y_B - \mu_B)^2 + \frac{1}{\sigma_C^2}(y_C - \mu_C)^2\right\}
\label{eqn:LRGeneral}
\end{equation}
\paragraph{A hypothesis with two equalities $A=B=C$.} This is the most simple case. The infimum in \ref{eqn:LRGeneral} is attained at the empirical mean of the observed data, that is
\begin{equation}
MLE = \left(\frac{y_A/\sigma_A^2+ y_B/\sigma_B^2+ y_C/\sigma_C^2}{1/\sigma_A^2+1/\sigma_B^2+1/\sigma_C^2},\frac{y_A/\sigma_A^2+ y_B/\sigma_B^2+ y_C/\sigma_C^2}{1/\sigma_A^2+1/\sigma_B^2+1/\sigma_C^2},\frac{y_A/\sigma_A^2+ y_B/\sigma_B^2+ y_C/\sigma_C^2}{1/\sigma_A^2+1/\sigma_B^2+1/\sigma_C^2}\right).
\label{eqn:MLEequal}
\end{equation}
The LR now has the form
\[LR = \sum_{i\in\{A,B,C\}}{\frac{(y_i - MLE_i)^2}{\sigma_i^2}}\]
and is distributed under the null hypothesis as a $\chi^2(2)$.\\
\paragraph{A hypothesis with two inequalities.} Let us consider the configuration $A<B<C$. We have 6 cases for the order of the observations which can be grouped into three subsets.
\begin{itemize}
\item A correct ordering, i.e. $y_A<y_B<y_C$. In this case, we have $LR=0$.
\item A totally incorrect ordering, i.e. $y_C<y_B<y_A$. In this case, the infimum in equation \ref{eqn:LRGeneral} is attained on the border of the null, that is when $\mu_A=\mu_B=\mu_C$ and the $LR$ is distributed under the hypothesis $\mu_A=\mu_B=\mu_C$ as $\chi^2(2)$.
\item A partially incorrect (or correct) ordering, i.e. $y_A<y_C<y_B$ or $y_B<y_A<y_C$. In this case one of the centers will keep its unconstrained value because it respects the ordering in the null hypothesis whereas the other centers need to be pooled. The LR is then distributed under the hypothesis $\mu_A=\mu_B=\mu_C$ as $\chi^2(1)$.
 \item In between a partial or total incorrect ordering, i.e $y_B<y_C<y_A,y_C<y_A<y_B$. In this case, we must distinguish between two situations. When $\frac{y_C/\sigma_C^2+y_A/\sigma_A^2}{1/\sigma_C^2 + 1/\sigma_A^2} < y_B$, then we get a $\chi^2(1)$. Otherwise, all the centers must be pooled and we get a $\chi^2(2)$.
\end{itemize}
By summing up these cases, we reach the same conclusion found in \cite{BartholomewPAVA}. Under the hypothesis $\mu_A=\mu_B=\mu_C$
\[\mathbb{P}_{\mu_A=\mu_B=\mu_C}\left(LR>\gamma\right) = P(2,3)\mathbb{P}(\chi^2(1)>\gamma) + P(1,3)\mathbb{P}(\chi^2(2)>\gamma),\]
where the numbers $P(1,3)$ and $P(2,3)$ are the probabilities corresponding to the different scenarios in points 2,3 and 4 for the relative positions of the observations with respect to the ordering imposed by the null hypothesis.

\paragraph{A hypothesis with one equality and one inequality.} Let us consider the configuration $A<B=C$. The LR is simplified into:
\begin{equation}
LR = \inf_{\mu_A<\mu_{BC}} \left\{\frac{1}{\sigma_A^2}(y_A - \mu_A)^2 + \frac{1}{\sigma_B^2}(y_B - \mu_{BC})^2 + \frac{1}{\sigma_C^2}(y_C - \mu_{BC})^2\right\}
\label{eqn:LROneEqOneIneq}
\end{equation}
The infimum in equation \ref{eqn:LROneEqOneIneq} cannot be calculated without the knowledge of the relative positions of the observations with respect to the constraint $\mu_A<\mu_{BC}$. Thus the calculus need to be done conditionally on it. More precisely, under no restriction in equation \ref{eqn:LROneEqOneIneq}, the infimum is attained when $\mu_A = y_A$ and $\mu_{BC}=\frac{ y_B/\sigma_B^2+ y_C/\sigma_C^2}{1/\sigma_B^2+1/\sigma_C^2}$. We have the following situations:
\begin{itemize}
\item if $y_A<\frac{ y_B/\sigma_B^2+ y_C/\sigma_C^2}{1/\sigma_B^2+1/\sigma_C^2}$, then the unconstrained optimum is inside the null hypothesis and thus is attained. Hence, the LR has the form
\[LR = \frac{1}{\sigma_B^2}\left(y_B - \frac{ y_B/\sigma_B^2+ y_C/\sigma_C^2}{1/\sigma_B^2+1/\sigma_C^2}\right)^2 + \frac{1}{\sigma_C^2}\left(y_C - \frac{ y_B/\sigma_B^2+ y_C/\sigma_C^2}{1/\sigma_B^2+1/\sigma_C^2}\right)^2\]
and conditionally on $y_A<\frac{ y_B/\sigma_B^2+ y_C/\sigma_C^2}{1/\sigma_B^2+1/\sigma_C^2}$, the LR has a $\chi^2(1)$ distribution under the hypothesis $\mu_A=\mu_B=\mu_C$.\\
\item if $y_A\geq\frac{ y_B/\sigma_B^2+ y_C/\sigma_C^2}{1/\sigma_B^2+1/\sigma_C^2}$, then the unconstrained optimum is not inside the null hypothesis and the optimum is thus attained on the border of the null hypothesis. In other words when $\mu_A=\mu_{BC}$. Thus the LR is given by
\[LR = \sum_{i\in\{A,B,C\}}{\frac{(y_i - MLE_i)^2}{\sigma^2}}\]
(MLE is defined by equation (\ref{eqn:MLEequal})) and is distributed conditionally on $y_A\geq\frac{ y_B/\sigma_B^2+ y_C/\sigma_C^2}{1/\sigma_B^2+1/\sigma_C^2}$ as a $\chi^2(2)$ under the hypothesis $\mu_A=\mu_B=\mu_C$. 
\end{itemize}
Since $y_A$ and $\frac{ y_B/\sigma_B^2+ y_C/\sigma_C^2}{1/\sigma_B^2+1/\sigma_C^2}$ are independent Gaussian random variables, the probability that either of them is greater than the other is equal to $1/2$ and we may conclude that under the hypothesis $\mu_A=\mu_B=\mu_C$
\begin{equation}
LR \;\sim\; \frac{1}{2} \chi^2(1) + \frac{1}{2} \chi^2(2).
\label{eqn:ProbaLawEqIneq3}
\end{equation}
For the case of more than 3 observations and an arbitrary number of equalities, we should be able to conclude in a similar way the probability distribution.

\section{A pseudo-code for the block-check algorithm described in paragraph \ref{subsec:BlocksApproach}}
\resizebox{!}{13cm}{
\begin{algorithm}[H]
\label{algo:BlockPartitioning}
\KwData{Ordered sample $y_1,\cdots,y_n$ and standard deviations $\sigma_1,\cdots,\sigma_n$. Vectors MinBlock and MaxBlock and the initial confidence intervals used to find MinBlock.
}
\KwResult{For each $i, [a_i,b_i]$ such that $\mathbb{P}(\forall i, \mu_i\in[a_i,b_i])\geq 1-\alpha$.}
\eIf {the top hypothesis is not rejected}
{Set confidence intervals to $[1,n]$; $\forall i, a_i=1, b_i = n$.}
{
 Set the confidence intervals to the initial confidence intervals.
 \For{$i$ from $1$ to $n-1$}
 {
	\If{MaxBlock$[i]\geq$(MinBlock$[i]+1$)}
	{
	  \For{k descending from MaxBlock[i] to (MinBlock[i]+1)}
		{
		  // Treat blocks of the form $\mu_a=\cdots=\mu_n$ or of the form $\mu_a=\cdots=\mu_{n-1}$.\;
		  \If{$k+i\geq n-1$}
			{
			  \If{the current block is not significant to the actual confidence interval}
				{
				  break the loop over $k$ because no further info can be gained for this center.
				}
			  \For{$j$ in partitions of $\{\mu_{k+1},\cdots,\mu_n\}$}
				{
				  Test the combination of the current block with partition number $j$ added to the left of it\;
					\If{the combination is not rejected}
					{
					  Update the confidence intervals\;
					  Break the loop over the partitions of $\{\mu_{k+1},\cdots,\mu_n\}$.
					}
				}
				\If{not all partitions in $\{\mu_{k+1},\cdots,\mu_n\}$ were tested}
				{
				  Break the loop over $k$ because some hypothesis was not rejected.
				}
			}
		 // Treat blocks of the form $\mu_a=\cdots=\mu_b$ with $a\geq 1$ and $b<n-1$.\;
		 \If{the current block is not significant to the actual confidence interval}
		 {
			  break the loop over $k$ because no further info can be gained for this center.
		 }
		 \For{$j$ in partitions of $\{\mu_{b+1},\cdots,\mu_n\}$}
		 {
		   Test the combination of the current block with partition number $j$ added to the right of it.\;
			 \If{the combination is not rejected}
			 {
				  Update the confidence intervals\;
				  Break the loop over the partitions of $\{\mu_{b+1},\cdots,\mu_n\}$.
			 }
			 \For{$s$ in partitions of $\{\mu_{1},\cdots,\mu_{a-1}\}$ in case there is a left part}
			 {
				 Test the combination of the current block with partition $j$ added to the right and partition $s$ to the left of it.\;
				 \If{the combination is not rejected}
				 {
					 Update the confidence intervals\;
					 Break the loop over the partitions of $\{\mu_{1},\cdots,\mu_{a-1}\}$.
				 }
			 }
			 \If{not all partitions in $\{\mu_{1},\cdots,\mu_{a-1}\}$ were tested}
				{
				  Break the loop over $j$ because some hypothesis was not rejected.
				}
		 }
		 \If{not all partitions in $\{\mu_{b+1},\cdots,\mu_n\}$ were tested}
		 {
			  Break the loop over $k$ because some hypothesis was not rejected.
		 }
		}

	}
 }
}
\caption{A block-based algorithm}
\end{algorithm}}

\bibliographystyle{plainnat}
\bibliography{biblioFile}

\end{document}